\documentclass[11pt]{article}
\pdfoutput=1
\usepackage{odonnell,booktabs,tabularx}
\allowdisplaybreaks
\usepackage{tocloft}
\usepackage{multirow}
\usepackage{colortbl}
\usepackage{bold-extra}
\renewcommand{\tilde}{\widetilde}
\renewcommand{\bar}{\overline}

\newcommand{\SIS}{\mathrm{SIS}^\infty}
\newcommand{\CIS}{\mathrm{CIS}}

\newcommand{\SubsetSum}{\textsc{Subset-Sum}\xspace}
\newcommand{\FSubsetSum}{$\F_3^n$\textsc{-}\SubsetSum}

\allowdisplaybreaks

\renewcommand{\backref}[1]{}

\renewcommand{\backrefalt}[4]{%
\ifcase #1 %
\or 
[p.\ #2]%
\else 
[pp.\ #2]%
\fi}

\title{No exponential quantum speedup for $\mathrm{SIS}^\infty$ anymore}
\author{
Robin Kothari\thanks{Google Quantum AI. Email: \texttt{robin@robinkothari.com}}
\and 
Ryan O'Donnell\thanks{Computer Science Department, Carnegie Mellon University.  Email: \texttt{odonnell@cs.cmu.edu}. Part of this work was done while consulting for Google Quantum AI.}
\and 
Kewen Wu\thanks{Institute for Advanced Study. Email: \texttt{shlw\_kevin@hotmail.com}. Part of the work was done while at Google.}
}

\begin{document}
\date{}
\maketitle

\begin{abstract}
In 2021, Chen, Liu, and Zhandry presented an efficient quantum algorithm for the average-case $\ell_\infty$-Short Integer Solution ($\mathrm{SIS}^\infty$) problem, in a parameter range outside the normal range of cryptographic interest, but still with no known efficient classical algorithm. This was particularly exciting since $\mathrm{SIS}^\infty$ is a simple problem without structure, and their algorithmic techniques were different from those used in prior exponential quantum speedups.

We present efficient classical algorithms for all of the $\mathrm{SIS}^\infty$ and (more general) Constrained Integer Solution problems studied in their paper, showing there is no exponential quantum speedup anymore.
\end{abstract}

\clearpage
\tableofcontents
\clearpage

\section{Introduction}\label{sec:intro}

Finding new problems where quantum algorithms yield an exponential speedup over classical ones remains a central challenge in quantum computer science.
We have long known a few classes of problems with a likely exponential quantum speedup, such as those based on the simulation of quantum systems or based on the hidden subgroup problem (e.g., integer factorization, discrete log).  Beyond those, it is rare and exciting to find a genuinely new class of problems with potential exponential quantum speedup, particularly if the problems are natural and previously studied.

In 2021, Chen, Liu, and Zhandry~\cite{chen2022quantum} did precisely this. They presented an efficient quantum algorithm for the well-known \emph{Short Integer Solution in $\ell_\infty$ norm} ($\SIS$) problem, in a less-studied parameter regime, but one with no known classical polynomial-time algorithm. Here is an example result they proved (see \Cref{thm:CLZSIS} for more). Let $q$ be a prime.\footnote{We use $q$ (instead of $p$) to be consistent with cryptography literature. Our results also work for prime power or composite modulus. See \Cref{sec:discussion_our} for details.}
\begin{quote}
    \textbf{Sample CLZ theorem.} \emph{There is a $\poly(m,q)$-time \emph{quantum} algorithm that, given a \emph{uniformly random} $H\in\F_q^{n\times m}$ where\footnote{For simplicity, we use $C$ to hide constants that may be different from place to place and may depend on other parameters described as ``constant''. In later sections we will formally state these constant dependencies.} $m\ge Cq^4\log q\cdot n^3$, finds a ``short'' nonzero $x \in \F_q^m$ satisfying $Hx = 0$, where ``short'' means $\|x\|_\infty < \lfloor q/2\rfloor$, i.e., $x_i \neq \pm\lfloor q/2\rfloor$ for all~$i$.}
\end{quote}
The particularly exciting part about this work is that it used genuinely different techniques from existing quantum speedups. Furthermore, in a more standard and more challenging parameter regime --- e.g., $m = C n$ and ``short'' meaning\footnote{Here $x$'s entries are interpreted in $\{-\lfloor q/2 \rfloor, \dots, \lfloor q/2 \rfloor\}$.} $\|x\|_\infty \leq q/4$ --- the $\SIS$ problem is of significant cryptographic interest. Variants of the problem underlie the security of several (candidate post-quantum) cryptographic systems, such as Dilithium~\cite{DKLLSSS18} and Wave~\cite{debris2019wave}; see \Cref{sec:discussion_crypto} for more details.

The Chen--Liu--Zhandry quantum algorithm is based on a quantum reduction by Regev~\cite{regev2009lattices}, who based the \emph{hardness} of the Learning With Errors problem on the hardness of finding short lattice vectors.
A similar reduction shows that $\SIS$ can be solved efficiently if a certain decoding problem can be solved efficiently.
Instead of using Regev's reduction to show hardness, Chen, Liu, and Zhandry devise a clever, efficient quantum algorithm for the decoding problem (in a certain parameter regime), and thereby obtain an efficient quantum algorithm for $\SIS$.

These \emph{algorithmic} applications of Regev's reduction have since attracted considerable interest~\cite{chen2022quantum,chailloux2023quantum,yamakawa2024verifiable,jordan2024optimization,chailloux2024quantum}. Besides the algorithm for $\SIS$, two other quantum algorithms that use the same primitive are:
(1)~the algorithm of Yamakawa and Zhandry~\cite{yamakawa2024verifiable} that achieves an exponential quantum--classical black-box separation for a search problem relative to a random oracle; and,
(2)~the Decoded Quantum Interferometry (DQI) algorithm for various optimization problems, such as the Optimal Polynomial Intersection (OPI) problem~\cite{jordan2024optimization,chailloux2024quantum}. 
The result by Yamakawa and Zhandry can be made non-oracular by instantiating the random oracle with a cryptographic hash function, providing a concrete problem that presumably retains the exponential quantum speedup. This leaves us with at least three non-oracular problems based on this algorithmic primitive with potential exponential speedups: $\SIS$, OPI, and Yamakawa--Zhandry.

Our main result is to dequantize\footnote{We do not simulate the CLZ algorithm in a classical way. Rather, we provide a classical algorithm directly.} \cite{chen2022quantum}. We show that $\SIS$, and a further generalization they study, can be solved by a classical algorithm whose efficiency even outperforms that of their quantum algorithm.
For example, we prove (see \Cref{thm:SIS} for more) the following theorem.
\begin{quote}
    \textbf{Sample new theorem.} \emph{There is a $\poly(m,\log q)$-time \emph{classical} algorithm that, given \emph{any} $H\in\F_q^{n\times m}$ where $m\ge Cn^3$, finds a ``short'' nonzero $x\in\F_q^m$ satisfying $Hx=0$, where ``short'' means $\|x\|_\infty \leq  \lfloor q/6 \rfloor$, i.e., $-\lfloor q/6\rfloor\le x_i\le\lfloor q/6\rfloor$ for all $i$.}
\end{quote}

We highlight several advantages of our result: (1) it is classical and deterministic; (2) it works for worst-case $H$; (3) the requirement on $m$ has no dependence on~$q$; (4) it is $\poly(n)$-time even when $q = 2^{\poly(n)}$; and (5) the notion of ``short'' is much stricter. (Our algorithm's classical running time is even faster than CLZ's quantum running time in the setting of comparable~$m$; see \Cref{sec:discussion_our}.)
As we also discuss in \Cref{sec:discussion_ivanyos}, for this particular sample theorem, a similar but slightly weaker version (without features (3)(4)) can be extracted from a recent work of Imran and Ivanyos~\cite{II24}.  

In general, the main theorem in this work gives a classical algorithm that outperforms the most general quantum algorithm from \cite{chen2022quantum}.  
This quantum algorithm is for a generalization of $\SIS$ we call Constrained Integer Solution ($\CIS$) problem, in which the entries of $x$ are restricted to lie in a general set $A$, as opposed to an interval. (For more details, see \Cref{sec:cis}.) 
\begin{theorem*}\textnormal{(\cite{chen2022quantum}){\bf.}}
    Let $k \geq 2$ be a constant and $A \subseteq \F_q$ be of size $|A| = q - k + 1$.
    There is a $\poly(m,q)$-time quantum algorithm that, given a uniformly random $H \in \F_q^{n \times m}$ where
    \begin{equation}
        m \geq C q^4 \log q \cdot n^k,
    \end{equation}
    finds a nonzero $x \in A^m$ satisfying $Hx = 0$. 
\end{theorem*}

Our main contribution is a classical algorithm that improves their quantum algorithm. Note that we not only improve their $q$-dependence in all cases, but improve their $n$-dependence in most cases.

\begin{theorem*}\textnormal{(Main)\textbf{.}} 
    There is a $\poly(m,q)$-time classical algorithm for the above problem, even for
    $$
    m\ge C\log q\cdot\begin{cases}
        n^2 & \text{whenever $q>4^{k-1}$}\\
        n^{k-1} & \text{whenever $k\ge3$ and $q\ge7$}\\
        n^k & \text{in general for all $k\ge2$ and $q\ge3$.}
    \end{cases}
    $$
\end{theorem*}

\subsection{A first warm-up: \texorpdfstring{\FSubsetSum}{F3-Subset-Sum}} \label{sec:f3}

The general $\SIS$ and $\CIS$ problems have several parameters, so we will warm up to a full description of the problem with some special cases.  The simplest special case is \FSubsetSum, which was already considered to have a potential exponential quantum speedup in~\cite{chen2022quantum}.

The \FSubsetSum problem is just a vector version of the classic \SubsetSum problem, but over $\F_3^n$ rather than the integers, and with the target fixed to~$0$:
\begin{center}
    Given $m$ vectors $h_1,\ldots,h_m \in \F_3^n$, find a nonempty subset $S \subseteq [m]$\footnote{For any integer $n\ge1$, we use $[n]$ to denote the set $\{1,2,\ldots,n\}$.} such that $\displaystyle \sum_{i\in S} h_i = 0$. 
\end{center}
In general, we can consider $\F_q^n$-\SubsetSum for any~$q$, and $q = 3$ is the first interesting case. 

The reader familiar with $\SIS$ will recognize \FSubsetSum as a worst-case version of it, in which $q=3$ and a vector $x \in \F_3^m$ is considered ``short'' if it is in $\{0,1\}^m$.
Indeed, as we will discuss in \Cref{sec:sis_related}, \FSubsetSum is (roughly) equivalent to a wide variety of problems: ``Binary-Error LWE with $q = 3$'', ``list-recovery for $\F_3$-linear codes'', ``ternary syndrome decoding with large (maximal) weight'', ``LIN-SAT over~$\F_3$'', ``Learning From Disequations over~$\F_3$'', and more.

As posed, the \FSubsetSum problem might not always have a solution, particularly if $m$ is small compared to~$n$.  But it is a mathematically nontrivial fact (discussed in \Cref{sec:related_zero-sum}) that once $m$ crosses the sharp threshold $2n$, a solution \emph{always} exists.
We will only study the problem when $m > 2n$, so there will always be a solution to find.

\FSubsetSum is also very interesting in the average-case setting, where the input vectors are chosen uniformly at random and the goal is to succeed on almost all inputs. Here it simple to show that a solution exists with overwhelming probability beyond threshold $(\log_2 3)n  \approx 1.58n$, a slightly smaller threshold than the one for worst-case existence.
It is this average-case version that is considered by Chen, Liu, and Zhandry.
Although they were mainly interested in $\SIS$ with larger~$q$, their work also contained the following interesting result.

\begin{theorem}[Special case of \protect{\cite[Remark 4]{chen2022quantum}}] \label{thm:CLZ3ss}
    The \FSubsetSum problem with uniformly random input vectors can be solved by a bounded-error $\poly(m)$-time quantum algorithm provided $m \geq C n^2$.\footnote{$C$ is not explicit in~\cite{chen2022quantum}. However, we can show (proof omitted) that $m \geq(n+1)(n+2)/2$ vectors suffice. We also remark that certain conjectures would imply $0 < C < 1/2$; but these are based on a generalization of the Arora--Ge algorithm~\cite{arora2011new} without rigorous performance guarantees; see, e.g.,~\cite{Ste24}.}
\end{theorem}

However, this problem can be solved efficiently and deterministically by a classical algorithm, even in the worst case. This was already shown implicitly in a 2007 paper of Ivanyos, Sanselme, and Santha: \cite[Claim 1]{ISS12} implies that \FSubsetSum can be solved for worst-case inputs when $m\geq(n+1)(n+2)/2\sim n^2/2$. We present an improved classical algorithm.

\begin{theorem}[Proved in \Cref{sec:F3-subset-sum}]\label{thm:intro_F3-subset-sum}
    The worst-case \FSubsetSum problem can be solved by a deterministic $\poly(m)$-time classical algorithm provided $m \geq n^2/3 + Cn \log n$. 
\end{theorem}

\subsection{Short Integer Solution in \texorpdfstring{$\ell_\infty$}{L-inf} norm, \texorpdfstring{$\SIS$}{SIS-inf}} 

We now discuss the main $\SIS$ problem studied by Chen, Liu, and Zhandry. 
Given a set of $m$ vectors over $\F_q^n$, the goal is to find a nontrivial linear combination of these that sums to $0$, where the vector $x$ of coefficients should have small infinity norm, meaning $x \in \{-s,\ldots,s\}^m$ for some positive integer $s$. 
Throughout this paper, we work with odd prime field $\F_q$ and view it as the set $\{-\lfloor q/2\rfloor,\ldots,\lfloor q/2\rfloor\}$. Thus the largest infinity norm of vectors in $\F_q^m$ is $\lfloor q/2\rfloor$. 

\begin{definition}
    In the $\SIS(n,m,q,s)$ problem for a prime $q > 3$, we are given a matrix $H \in \F_q^{n\times m}$ and an integer $s \in \{1,\ldots, \lfloor q/2\rfloor\}$. The goal is to find a nonzero vector $x\in \F_q^m$ with $\|x\|_\infty\leq s$ satisfying $Hx = 0$.
\end{definition}

This problem gets easier as $m$ gets larger and as $s$ gets closer to $\lfloor q/2\rfloor$. For example, when $s=\lfloor q/2\rfloor$, every nonzero vector is allowed and we only need $m=n+1$ vectors.

This problem is interesting both in the worst case and the average case.\footnote{Assuming worst-case hardness of some lattice problems and passing through the SIS problem in $\ell_2$ norm \cite{Ajt96,micciancio2007worst}, the average-case $\SIS$ problem is hard when $s=O(1)$, $q=\Theta(n^2\log n)$, and $m=\Theta(n\log n)$.} The main result of Chen, Liu, and Zhandry is a quantum algorithm for the problem in the average case, where the matrix $H$ is drawn uniformly at random from $\F_q^{n\times m}$.

\begin{theorem}[\protect{\cite[Theorem 2]{chen2022quantum}}]\label{thm:CLZSIS}
    For any odd constant $k \geq 3$, there is a bounded-error $\poly(m,q)$-time quantum algorithm for \emph{average-case} $\SIS$ when
    \begin{equation}
    s=\frac{q-k}{2}~~\text{and}~~m \geq C q^4\log q \cdot n^{k}.
    \end{equation}
\end{theorem}

This algorithm works for average-case inputs and is only $\poly(n)$-time when $q\leq \poly(n)$. The main result in our work is a classical algorithm that works for worst-case inputs, achieves a significantly smaller~$s$, and allows $q$ to be exponentially large in~$n$.

\begin{theorem}[Proved in \Cref{sec:general_cis_sis}]\label{thm:SIS}
    For any constant $k \geq 2$, there is a deterministic\\ $\poly(m,\log q)$-time classical algorithm for \emph{worst-case} $\SIS$ when
    \begin{equation}
        s = \Bigl\lfloor\frac{q}{2k}\Bigr\rfloor~~\text{and}~~m \geq Cn^k.
    \end{equation}
\end{theorem}

We also have a variant of \Cref{thm:SIS} which has an additional assumption on $k$, but a substantially simpler proof and better control on the constants.
\begin{theorem}[Simpler version; proved in \Cref{sec:simple_sis}]\label{thm:simple_halving_sis_intro}
    For any $k \leq \lfloor q/2\rfloor$ that is a positive-integer power of~$2$, there is a deterministic $\poly(m,\log q)$-time classical algorithm for \emph{worst-case} $\SIS$ when 
    \begin{equation}
        s = \Bigl\lfloor\frac{q}{2k}\Bigr\rfloor~~\text{and}~~m 
        \geq \left(\frac{n+1}{c_q}\right)^{k},
    \end{equation}
    where $c_q = \sqrt{2 + \frac{2}{q-1} - o(1)}\ge1$.
\end{theorem}
We present the proof of this weaker theorem in \Cref{sec:halving_trick}, as a warm-up to the more elaborate proof needed for \Cref{thm:SIS}.  As we discuss in \Cref{sec:discussion_ivanyos}, this easier proof is based on a method of Imran and Ivanyos~\cite{II24}, but improves upon it by allowing exponentially large~$q$.

\subsection{\texorpdfstring{$\F_q^n$-\SubsetSum}{Fq-Subset-Sum}}

It is natural also to study the $\F_q^n$-\SubsetSum problem, which is like $\SIS$ with $s = 1$, but even harder: rather than seeking a nonzero $x \in \{-1,0,1\}^m$ with $Hx = 0$, we seek $x \in \{0,1\}^m$.  In order to achieve $m = \poly(n)$, it seems necessary to focus on constant $q$. We also exclude $q = 3$, which is already studied in \Cref{sec:f3}.

There are two prior algorithms  we can compare against: a quantum algorithm by Chen, Liu, and Zhandry \cite{chen2022quantum} and a classical one by Imran and Ivanyos \cite{II24}.

\begin{theorem}[Follows from \protect{\cite[Remark 4]{chen2022quantum}}]  \label{thm:clzss}
    For constant prime $q > 3$, the $\F_q^n$-\SubsetSum problem with $m$ uniformly random input vectors can be solved by a bounded-error $\poly(m)$-time quantum algorithm provided $m \geq C n^{q-1}$.
\end{theorem}
\begin{theorem}[Follows from \protect{\cite[Proposition 4]{II24}}]  \label{thm:iiss}
    For constant prime $q > 3$, the $\F_q^n$-\SubsetSum problem with $m$ \emph{worst-case} input vectors can be solved by a deterministic $\poly(m)$-time classical algorithm provided $m \geq C n^{\frac14 \ol{q} \log_2 \ol{q}}$, where $\ol{q}$ is $q$ rounded up to the next integer power of~$2$.  
    
    Note that the exponent here is $\Omega(q\log q)$, which exceeds $q-1$ for $q \neq 7$.
\end{theorem}

We prove the following result for $\F_q^n$-\SubsetSum, which strictly improves \Cref{thm:clzss}.
Note that our result is incomparable with \Cref{thm:iiss}, as we have a noticeably better exponent, but works in the average case.

\begin{theorem}[Proved in \Cref{sec:general_cis_subset-sum}]\label{thm:FpSS}
    For constant prime $q > 3$, the $\F_q^n$-\SubsetSum problem with $m$ uniformly random input vectors can be solved by a deterministic $\poly(m)$-time classical algorithm provided $m  \geq Cn^{\left\lceil\! \lceil q/2 \rceil\!\right\rceil}$, where $\left\lceil\!\lceil q/2 \rceil\!\right\rceil$ denotes the smallest even integer exceeding~$q/2$.  
    
    Note that the exponent is $q/2+o(1)$, which is always at most $q-1$ and strictly so for $q\neq 5$.
\end{theorem}
Once again, the proof of \Cref{thm:FpSS} is somewhat involved, and we present a less complicated version as \Cref{thm:subset-sum_avg}, for $m \geq C n^{\overline{q-1}}$, where $\overline{q-1}$ is $q-1$ rounded up to the next power of~$2$. 

\subsection{Constrained Integer Solution (\texorpdfstring{$\CIS$}{CIS})} \label{sec:cis}

While the main result presented by Chen, Liu, and Zhandry is about $\SIS$, their algorithm works for a more general problem, as pointed out in \cite[Remark 4]{chen2022quantum}.
In this version, instead of requiring that the solution $x$ has small infinity norm, we require that each entry of $x$ belongs to a specified allowed set $A \subseteq \F_q$.\footnote{We remark that the results in \cite{chen2022quantum} work for the more general setting where each entry of $x$ gets its own restricted subset. Our algorithms also work in this setting; see \Cref{sec:discussion_our} for a discussion.}
We call this problem ``Constrained Integer Solution ($\CIS$)'' or sometimes $A$-$\CIS$ to emphasize the set $A$.

Note that setting $A = \{-s,\ldots,s\}$ recovers the $\SIS$ problem; and setting $A=\{0,1\}$ recovers the $\SubsetSum$ problem. We also describe how $\CIS$ unifies Yamakawa-Zhandry \cite{yamakawa2024verifiable} and OPI \cite{jordan2024optimization} in \Cref{sec:discussion_regev}.

\begin{definition}
    In the $\CIS(n,m,q,A)$ problem for a prime $q\ge3$, we are given a matrix $H \in \F_q^{n\times m}$ and a set $A \subseteq \F_q$ of size $2\le|A|\le q-1$.
    The goal is to find a nonzero vector $x\in A^m$ satisfying $Hx = 0$.
\end{definition}

Like \cite{chen2022quantum}, we will only study this problem in the average case, where $H$ is a uniformly random matrix in $\F_q^{n \times m}$, as the worst-case instances may not guarantee solutions.
We also remark that the problem seems quite hard unless $|A|$ is close to~$q$; thus we parameterize $|A| = q- k+1$. 

It turns out that the same quantum algorithm in \cite{chen2022quantum} that solves $\SIS$ can be used for $\CIS$ with similar bounds.

\begin{theorem}[\protect{\cite[Remark 4]{chen2022quantum}}]\label{thm:CLZCIS}
    Let $k \geq 2$ be a constant and $A \subseteq \F_q$ be of size $|A| = q - k + 1$.
    There is a bounded-error $\poly(m,q)$-time quantum algorithm for average-case $A$-$\CIS$ when 
    \begin{equation}
        m \geq C q^4 \log q \cdot n^k.
    \end{equation}
\end{theorem}

We give classical algorithms for the same problem.  First, we have a simple argument to handle the case when $q$ is much larger than $k$. This is in fact the motivating setting in \cite{chen2022quantum}, as they typically describe $q$ being polynomial in~$n$ and $k$ being constant.

\begin{theorem}[Proved in \Cref{sec:simple_cis}]\label{thm:simple_cis_lev_intro}
    Let $k \geq 2$ be a constant and $A \subseteq \F_q$ be of size $|A| = q - k + 1$.
    Assume $q>4^{k-1}$.
    There is a deterministic $\poly(m,q)$-time classical algorithm for average-case $A$-$\CIS$ when 
    \begin{equation}
        m\geq C \log q \cdot n^2.
    \end{equation}
\end{theorem}
\noindent Note here that the exponent on $n$ is fixed to~$2$; it does not grow with~$k$ as in \Cref{thm:CLZCIS}.
\medskip

On the other hand, for general $q$, we have another, more involved, classical algorithm.
Compared with \Cref{thm:CLZCIS}, our requirement on $m$ always has much better $q$-dependence and has better $n$-dependence in almost all settings.

\begin{theorem}[Proved in \Cref{sec:general_cis_cis}]\label{thm:CIS_intro}
For any constant $k\ge3$ and prime $q>5$, given a set $A \subseteq \F_q$ with $|A| = q-k+1$, there is a deterministic $\poly(m,q)$-time classical algorithm for average-case $A$-$\CIS$ with 
\begin{equation}
    m \geq C\log q\cdot n^{k-1}.
\end{equation}
The above bound also holds when $k=3$ and $q=5$.

In the remaining cases of $k = 2$ or $k=4,q=5$, the algorithm requires $m \geq C\log q\cdot n^k$.
\end{theorem}

Our focus is to provide \emph{some} bound that matches or improves the known quantum algorithms; and we do not make the effort to exhaust the \emph{best possible} saving in all parameter regimes.

\subsection{Results summary}

In \Cref{tab:result_summary}, we compare efficient algorithms for $\SIS$-type problems. 
For simplicity, we use $\eps$ to hide $o(1)$ factors that go to zero as $n$ or $q$ gets larger. Footnote-sized conditions inside parentheses are additional assumptions of the problem or algorithm.

\begin{table}[h!]
\caption{Comparison of efficient algorithms for $\SIS$-type problems.}\label{tab:result_summary}
\centering
\begin{tabular}{l@{\hspace{20pt}}l@{\hspace{20pt}}l@{\hspace{20pt}}l}
\hline\\[-1em]
\bf Problem\ \textbackslash\ Work
& 
\begin{tabular}[c]{@{}l@{}}\cite{chen2022quantum}\\
\bf (quantum)\end{tabular} 
& 
\begin{tabular}[c]{@{}l@{}}\cite{ISS12,II24}\\
\bf (classical)\end{tabular}
& 
\begin{tabular}[c]{@{}l@{}}
\bf Present work\\
\bf (classical)\end{tabular}
\\[-1em]\\\hline\\[-1em]
\begin{tabular}[c]{@{}l@{}}$A$-$\CIS$\\$|A|=q-k+1$\\\\
\footnotesize (constant $k$)\\\footnotesize (average-case)\end{tabular}
& 
$m\ge Cq^4\log q\cdot n^k$
& 

& 
\begin{tabular}[c]{@{}l@{}}
$m\ge C\log q\cdot n^2$\\[-1em]\\\\[-1em]
\footnotesize ($q>4^{k-1}$)
\\\\[-1em]\hline\\[-1em]
$m\ge C\log q\cdot n^k$\\[-1em]\\\\[-1em]
\footnotesize (general $q$)\end{tabular} 
\\[-1em]\\\hline\\[-1em]
\FSubsetSum
& 
\begin{tabular}[c]{@{}l@{}}$m\ge Cn^2$
\\[-1em]\\\\[-1em]
\footnotesize (average-case)\end{tabular}
& 
\begin{tabular}[c]{@{}l@{}}$m\ge(1/2+\eps)n^2$
\\[-1em]\\\\[-1em]
\footnotesize (worst-case)\end{tabular}
& 
\begin{tabular}[c]{@{}l@{}}$m\ge(1/3+\eps)n^2$
\\[-1em]\\\\[-1em]
\footnotesize (worst-case)\end{tabular}     
\\[-1em]\\\hline\\[-1em]
\begin{tabular}[c]{@{}l@{}}$\F_q^n$-\SubsetSum\\\\
\footnotesize (constant $q$)\end{tabular}
& \begin{tabular}[c]{@{}l@{}}$m\ge Cn^{q-1}$
\\[-1em]\\\\[-1em]
\footnotesize (average-case)\end{tabular}
& 
\begin{tabular}[c]{@{}l@{}}$m\ge n^{C q\log q}$
\\[-1em]\\\\[-1em]
\footnotesize (worst-case)\end{tabular}
& 
\begin{tabular}[c]{@{}l@{}}$m\ge Cn^{q/2+\eps}$
\\[-1em]\\\\[-1em]
\footnotesize (average-case)\end{tabular}
\\[-1em]\\\hline\\[-1em]
\begin{tabular}[c]{@{}l@{}}$\SIS$\\\\
\footnotesize (constant $k$)\end{tabular}                   
& 
\begin{tabular}[c]{@{}l@{}}$s=(q-k)/2$
\\[-1em]\\\\[-1em]
$m\ge Cq^4\log q\cdot n^k$
\\[-1em]\\\\[-1em]
\footnotesize (odd $k$)\\\footnotesize (average-case)\end{tabular} 
& 
\begin{tabular}[c]{@{}l@{}}$s=\lfloor q/(2k)\rfloor$
\\[-1em]\\\\[-1em]
$m\ge q^{Ck\log k}\cdot n^k$
\\[-1em]\\\\[-1em]
\footnotesize ($k$ power of $2$)\\\footnotesize (worst-case)\end{tabular} 
& 
\begin{tabular}[c]{@{}l@{}}$s=\lfloor q/(2k)\rfloor$
\\[-1em]\\\\[-1em]
$m\ge Cn^k$\\ \ 
\\[-1em]\\\\[-1em]
\footnotesize (worst-case)\end{tabular}
\\[-1em]\\\hline
\end{tabular}
\end{table}

\paragraph{Paper organization.}
In \Cref{sec:F3-subset-sum}, we present our algorithms on the \FSubsetSum problem.
In \Cref{sec:halving_trick}, we describe a simple halving trick that improves upon the techniques in \cite{II24} to reduce the weight of solutions, which leads to simple algorithms for $\SIS$, $\SubsetSum$, and $\CIS$.
Then in \Cref{sec:general_cis}, we develop extra techniques to improve previous algorithms.
Finally in \Cref{sec:discussion}, we mention further improvements on our results and discuss problems, algorithms, and motivations related to our work.


\section{\texorpdfstring{\FSubsetSum}{F3-Subset-Sum} algorithms}\label{sec:F3-subset-sum}

The goal of this section is to handle the simplest field $\F_3=\{-1,0,1\}$ and prove \Cref{thm:intro_F3-subset-sum}.

We start with a weaker version of \Cref{thm:intro_F3-subset-sum} with the simplest proof. Then we identify some room for improvement and finally obtain \Cref{thm:intro_F3-subset-sum}.
We will use zero-sum to refer to a linear combination of vectors that equals zero; and we say it is nontrivial if it uses some nonzero coefficient in the linear combination.

\begin{theorem} \label{thm:f3-zero-sum-weak}
The worst-case \FSubsetSum problem can be solved in deterministic $\poly(m)$ time when $m\ge(n+1)^2$.
\end{theorem}
\begin{proof}
Take the set $S_1$ of the first $n + 1$ vectors and find a nontrivial linear combination of them equaling $0$. 
Collecting like coefficients together, this yields $u_1 - u_1' =0$, where $u_1$ and $u_1'$ are vectors formed by disjoint subset-sums from~$S_1$ and they are not both empty.
Note that $u_1 + u_1' = u_1+u_1 = -u_1$, so $-u_1$ is also a subset-sum from~$S_1$.

Repeat for subsequent sets of $n+1$ vectors, $S_2, \dots, S_{n+ 1}$, producing $u_2, \dots, u_{n+1}$.
For each $i \in [n+1]$, we have that $\pm u_i$ is a subset-sum from~$S_i$. 

Finally, find a nontrivial linear combination of $u_1, \dots, u_{n+1}$ equaling $0$.
Collecting like coefficients, this yields $v - v' = 0$, where $v$ and $v'$ are subset-sums of the $u_i$'s, not both empty. 
But now $v$ is a subset-sum of the original vectors (since the $u_i$'s are), and so too is $-v'$ (since the $-u_i$'s are).  
This gives a nontrivial zero-sum from the original vectors as desired.
\end{proof}

One way to improve \Cref{thm:f3-zero-sum-weak} is to observe a dimension reduction: once $u_1$ is created, we can project later vectors onto the space orthogonal to $u_1$; this is because $0\cdot u_1,1\cdot u_1,-1\cdot u_1$ can all be replaced by a subset-sum from $S_1$.
This leads to the following \Cref{thm:f3-zero-sum-less-weak}, which is a special case of \cite[Claim~1]{ISS12} and we present a short proof here.

\begin{theorem}\label{thm:f3-zero-sum-less-weak}
The worst-case \FSubsetSum problem can be solved in deterministic $\poly(m)$ time when $m\ge(n+1)(n+2)/2$.
\end{theorem}
\begin{proof}
It will be more convenient to work with the following more general statement: given $m\ge(\ell+1)(\ell+2)/2$ vectors in a linear subspace $V$ of $\F_3^n$ with dimension $\ell$, we can find in polynomial time a nontrivial zero-sum.
Then \Cref{thm:f3-zero-sum-less-weak} follows by setting $\ell=n$ and $V=\F_3^n$.

We prove the above general statement by induction on $\ell$ and the runtime will be clear from the analysis.
The base case of $\ell = 0$ is trivial, as the given vector must be $0$.
For the induction step, let $S$ be the set of  first $\ell +1$ vectors and $T$ the last $\ell(\ell+1)/2$.
The vectors in $S$ must be linearly dependent, so we can find a nontrivial linear combination of them that sums to $0$.
Write this linear combination as $u - u' = 0$, where $u, u'$ are disjoint subset-sums from~$S$ and are not both empty.
Similar to the proof of \Cref{thm:f3-zero-sum-weak}, $-u=u+u'$ is also a subset-sum from~$S$.
    
Now if $u =0$ then we are done.  
Otherwise write $V = \spn\{u\} \oplus W$, where $W$ has dimension $\ell-1$.
Each vector $v_i \in T$ can be written as $c_i u + w_i$ for some $c_i \in \F_3$ and $w_i \in W$. Applying induction to the $w_i$'s gives a nontrivial zero-sum, i.e., a nonempty collection $\{w_i\colon i \in I\subseteq T\}$ with $\sum_{i \in I} w_i =0$.  
Now the associated collection of original vectors $\{v_i\colon i \in I\}$ sums to $\sum_{i\in I}v_i=c\cdot u$ for $c=\sum_{i\in I}c_i\in\F_3$.
By our construction of $u$ above, $-c\cdot u$ can be expressed as a subset-sum from $S'\subseteq S$.
This means $\sum_{i\in S'\cup I}v_i=\sum_{i\in I}v_i-c\cdot u=0$ is a desired nontrivial zero-sum.
\end{proof}

To further optimize parameters, we observe that each time the dimension reduction relies on a nontrivial linear equation. While both \Cref{thm:f3-zero-sum-weak} and \Cref{thm:f3-zero-sum-less-weak} generate such equations by the simple rank$+1$ bound, we can go beyond this, thanks to the following result.

\begin{lemma}\label{lem:f3-zero-sum-zeros}
Let $\ell\ge0$ be an integer and $t=\lceil\log_3(\ell+1)\rceil$.
Given $m\ge\ell+t$ vectors in a linear subspace $V$ of $\F_3^n$ with dimension $\ell$, we can find in deterministic $\poly(n,m)$ time a nontrivial zero-sum that uses at most $2(\ell+1)/3+t$ vectors.
\end{lemma}
\begin{proof}
Without loss of generality we assume the input vectors span the full subspace $V$, since otherwise we can always work with the smaller subspace and fewer vectors.
By taking an invertible linear transform, we assume the input vector are standard basis vectors $e_1,\ldots,e_\ell$ and additional vectors $v_1,\ldots,v_t$.

For each $\alpha=(\alpha_1,\ldots,\alpha_t)\in\F_3^t$, define $v_\alpha=\alpha_1v_1+\cdots+\alpha_tv_t$.
Now it suffices to find some $\alpha\ne0^t$ such that $v_\alpha$ has at most $2(\ell+1)/3$ nonzero entries, in which case these nonzero entries can be canceled by $2(\ell+1)/3$ standard basis vectors to obtain the desired zero-sum.

To find such an $\alpha$, we consider a uniform $\alpha\in\F_3^t$.
Then the vector $v_\alpha$ has at most $2\ell/3$ nonzeros in expectation; this is because every coordinate of $v_\alpha$ is either constantly zero or uniform in $\F_3$.
Since the trivial choice $\alpha=0^t$ happens with probability $3^{-t}$, there is some $\alpha\ne0^t$ such that $v_\alpha$ has at most 
$$
\frac{2\ell/3}{1-3^{-t}}\le\frac{2\ell/3}{1-\frac1{\ell+1}}=\frac{2(\ell+1)}3
$$ 
nonzeros and $\alpha$ can be efficiently found by brute-force enumeration.
\end{proof}

We remark that the above lemma (and the upcoming generalization \Cref{lem:sis_halving_sparse}) can be viewed as a coding-theoretic result to efficiently find a sparse codeword. Along this line, similar result can be derived from algorithmizing the Plotkin bound \cite{plotkin1960binary}, though they are stuck at the same $2/3$ ratio.

Now we are ready to prove \Cref{thm:intro_F3-subset-sum}.
\begin{proof}[Proof of \Cref{thm:intro_F3-subset-sum}]
We follow the same strategy of the inductive arguments in \Cref{thm:f3-zero-sum-less-weak}. The only difference is we use \Cref{lem:f3-zero-sum-zeros} to find the set $S$ of linear dependent vectors. The detailed calculation is as follows.

Let $m_0=1$ and 
\begin{equation}\label{eq:thm:intro_F3-subset-sum_1}
m_\ell=\max\left\{\ell,m_{\ell-1}+\frac{2(\ell+1)}3\right\}+\lceil\log_3(\ell+1)\rceil.
\end{equation}
We will prove the statement: given $m\ge m_\ell$ vectors in a linear subspace $V$ of $\F_3^n$ with dimension $\ell$, we can find in polynomial time a nontrivial zero-sum. 
The base case $\ell=0$ is again trivial.
Now consider the induction step at dimension $\ell\ge1$.
Define $t=\lceil\log_3(\ell+1)\rceil$.
Since we have $m\ge m_\ell\ge\ell+t$ vectors, we can use \Cref{lem:f3-zero-sum-zeros} to find a set of at most $m'=2(\ell+1)/3+t$ vectors with linear dependence; then apply induction hypothesis for dimension $\ell-1$ with the remaining $m-m'\ge m_{\ell-1}$ vectors, similar to the proof of \Cref{thm:f3-zero-sum-less-weak}.

Finally, to analyze \Cref{eq:thm:intro_F3-subset-sum_1}, we define 
$$
\bar m_\ell=\frac{(\ell+1)(\ell+2)+1}3+\ell\cdot\lceil\log_3(\ell+1)\rceil.
$$
Then a simple inductive argument proves $m_\ell\le\bar m_\ell$: the base case $\ell=0$ is trivial; the inductive case $\ell\ge1$ is proved by observing $\ell+\lceil\log_3(\ell+1)\rceil\le\ell\cdot\lceil\log_3(\ell+1)\rceil\le\bar m_\ell$ and
$$
m_{\ell-1}+\frac{2(\ell+1)}3+\lceil\log_3(\ell+1)\rceil
\le\bar m_{\ell-1}+\frac{2(\ell+1)}3+\lceil\log_3(\ell+1)\rceil\le\bar m_\ell.
$$
This completes the proof of \Cref{thm:intro_F3-subset-sum} by setting $\ell=n$ and $V=\F_3^n$.
\end{proof}


\section{A halving trick and its applications}\label{sec:halving_trick}

The main feature we explore in \Cref{sec:F3-subset-sum} is to use linear dependence to construct a subset-sum $u\in\F_3^n$ such that $0\cdot u,1\cdot u,-1\cdot u$ can all be expressed as subset-sums.
For larger fields and considering signs, this can be generalized as the following definition of \emph{reducible} vector, which will be helpful in designing algorithms for $\SIS$, $\SubsetSum$, and $\CIS$.

The goal of this section is to give a primitive exposition of the halving trick using reducible vectors.
We note that our halving trick is inspired by a similar one in \cite{II24}, and our main contribution here is to remove the dependence on the field size $q$ in \cite{II24}.

\begin{definition}[Bounded Vector Sum]\label{def:vector-sum_sis}
Let $v_1,\ldots,v_m$ be vectors and let $h\ge1$ be an integer.
We say vector $u$ is a $(\pm h)$-sum of the $v_i$'s if we can choose integer $-h\le\alpha_i\le h$ for each $i\in[m]$ such that $u=\sum_{i\in[m]}\alpha_iv_i$; and it is a $(\pm h)$-zero-sum if $u=0$.
We say the sum is nontrivial if $\alpha_i$'s are not all zero.
\end{definition}

\begin{definition}[Reducible Vector]\label{def:reducible_sis}
Let $v_1, \dots, v_m$ be vectors and let $1 \leq h' < h$ be integers. 
We say that vector $u$ is $(\pm h \to\pm h')$-reducible for the $v_i$'s if
\begin{itemize}
    \item $u$ is a nontrivial $(\pm1)$-sum of the $v_i$'s;
    \item for every integer $-h \leq c \leq h$, the vector $c\cdot u$ is a $(\pm h')$-sum of the~$v_i$'s.
\end{itemize}
\end{definition}

For clarity, our presentation in this section will focus on the existence of vector sums with certain properties (e.g., a solution to $\SIS$) and the runtime of such a construction will be obvious from the proof and thus omitted.

We begin with showing that zero-sums give reducible vectors.

\begin{lemma}\label{lem:reducible_sis_halving}
Let $h > 1$ be an integer. 
Given a nontrivial $(\pm h)$-zero-sum of vectors $v_1, \dots, v_m$,
\begin{equation}\label{eq:lem:reducible_sis_halving_1}
\alpha_1 v_1 + \cdots + \alpha_m v_m = 0
\quad\text{where each }
-h \leq \alpha_i \leq h
\text{ is an integer},
\end{equation}
we can find in deterministic $\poly(m,\log h)$ time a $(\pm h\to \pm\lfloor h/2\rfloor)$-reducible~$u$ for the $v_i$'s.
\end{lemma}
\begin{proof}
Define $a=\max_{i\in[m]}|\alpha_i|$, which satisfies $1\le a\le h$ since \Cref{eq:lem:reducible_sis_halving_1} is a nontrivial $(\pm h)$-zero-sum.
Define $h'=\lfloor h/2\rfloor$.
We first handle the simpler case that $a>h'$.
In this case, we construct $u$ by
$$
u=\sum_{i\in[m]\colon\alpha_i>h'}v_i-\sum_{i\in[m]\colon\alpha_i<-h'}v_i,
$$
which is a nontrivial $(\pm1)$-sum since $a>h'$.
In addition for any $0\le c\le h$, if $c\le h'$, then $c\cdot u$ is a $(\pm c)$-sum, which is immediately a $(\pm h')$-sum; if $h'<c\le h$, then we express $c\cdot u$ by
$$
c\cdot u=c\cdot u-\sum_{i\in[m]}\alpha_iv_i
=\sum_{i\colon|\alpha_i|\le h'}(-\alpha_i)v_i+\sum_{i\colon\alpha_i>h'}(c-\alpha_i)v_i+\sum_{i\colon\alpha_i<-h'}-(c+\alpha_i)v_i,
$$
which is a $(\pm h')$-sum since $|\alpha_i|\le h$ and $h'=\lfloor h/2\rfloor$. Similar calculation works for $-h\le c\le0$.

For the other case that $1\le a\le h'$, we compute $t=\lfloor h/a\rfloor$ and consider the following equivalent form of \Cref{eq:lem:reducible_sis_halving_1}:
$$
\alpha'_1v_1+\cdots+\alpha'_mv_m=0
\quad\text{where }\alpha'_i=t\cdot\alpha_i\text{ satisfies }-h\le\alpha'_i\le h.
$$
Since $1\le a\le h'=\lfloor h/2\rfloor$, we know $t=\lfloor h/a\rfloor>h/(2a)$ and hence $\max_{i\in[m]}|\alpha'_i|=t\cdot a>h/2\ge h'$.
Therefore this reduces to the above case.

Finally we remark that the procedure above is oblivious to the vectors $v_1, \dots, v_m$; and it only works on the coefficients $\alpha_1, \dots, \alpha_m$ to construct $u$ and express each $c\cdot u$ as a $(\pm\lfloor h/2\rfloor)$-sum.
Therefore it has runtime $\poly(m,\log h)$ that is independent of $n$.
\end{proof}

\Cref{lem:reducible_sis_halving} leads to the following corollary, which halves the weight of the $\SIS$ solutions at the expense of squaring.

\begin{corollary}\label{cor:iterate_sis_halving}
Let $q$ be a prime and let $h\ge2$ be an integer.
Suppose there is a deterministic algorithm $\calA$, given $m$ input vectors in $\F_q^n$ and running in time $T$, that outputs a nontrivial $(\pm h)$-zero-sum.

Then there is a deterministic algorithm $\calB$, given $m^2$ vectors in $\F_q^n$ and running in time $(m + 1)T + \poly(m,\log h)$, that outputs a nontrivial $(\pm\lfloor h/2\rfloor)$-zero-sum.
\end{corollary}
\begin{proof}
The algorithm $\calB$ will split the input vectors $v_1, \dots, v_{m^2}$ into $m$ batches $S_1, \dots, S_m$ of size $m$. 
For each $i\in[m]$, $\calB$ applies $\calA$ to the vectors in $S_i$ to get a nontrivial $(\pm h)$-zero-sum from $S_i$, and then applies \Cref{lem:reducible_sis_halving} to get a $(\pm h \to\pm\lfloor h/2\rfloor)$-reducible $u^{(i)}$ for $S_i$.
Note that if any $u^{(i)}=0$, then we are already done by outputting $u^{(i)}$.

Assume $u^{(1)},\ldots,u^{(m)}$ are all nonzero vectors.
Then $\calB$ applies $\calA$ again, to $u^{(1)}, \dots, u^{(m)}$, thereby getting a nontrivial $(\pm h)$-zero-sum $\sum_i c_i u^{(i)} = 0$, where the $c_i$'s are not all zero.
Since each $u^{(i)}$ is $(\pm h \to \pm\lfloor h/2\rfloor)$-reducible, we can replace $c_i u^{(i)}$ with a $(\pm\lfloor h/2\rfloor)$-sum of vectors from $S_i$.
After the substitution, we obtain a $(\pm\lfloor h/2\rfloor)$-zero-sum.
To see it is nontrivial, we recall that some $c_i$ is nonzero and the corresponding $u^{(i)}$ is also nonzero, which means $c_iu^{(i)}$ is nonzero and hence its substitution cannot be a trivial sum.

Finally we remark that the runtime involves $m+1$ executions of $\calA$ with time $T$ each, and $m$ executions of \Cref{lem:reducible_sis_halving} with time $\poly(m,\log h)$ each.
\end{proof}

\subsection{Simple \texorpdfstring{$\SIS$}{SIS-inf} algorithms}\label{sec:simple_sis}

An immediate starting point for \Cref{cor:iterate_sis_halving} is the simple algorithm to find a nontrivial linear dependence in $n+1$ vectors.
This leads to the following theorem for the $\SIS$ problem.

\begin{theorem}\label{thm:halving_sis_weak}
Let $q\ge3$ be a prime and $1\le k\le\lfloor q/2\rfloor$ be an integer power of $2$.
Then there is a deterministic $\poly(m,\log q)$ time algorithm, given $m=(n+1)^k$ input vectors in $\F_q^n$, that outputs a nontrivial $(\pm\lfloor q/(2k)\rfloor)$-zero-sum.
\end{theorem}
\begin{proof}
The base case $k=1$ is the simple algorithm to find a nontrivial $(\pm\lfloor q/2\rfloor)$-zero-sum in $n+1$ vectors, which runs in time $\poly(n,\log q)$.
The inductive case $k\ge2$ follows from \Cref{cor:iterate_sis_halving} with the inductive hypothesis for $k/2$, and noticing $\lfloor\lfloor q/k\rfloor/2\rfloor=\lfloor q/(2k)\rfloor$.
\end{proof}

To improve \Cref{thm:halving_sis_weak}, we explore the same ideas of dimension reduction and sparser zero-sum in \Cref{sec:F3-subset-sum} to get a better starting point.
We begin with the idea of finding sparser zero-sums.

\begin{lemma}\label{lem:sis_halving_sparse}
Let $r\ge1$ be an arbitrary integer.
Given $\ell+r$ vectors in a linear subspace of $\F_q^n$ with dimension $\ell$, we can find in deterministic $\poly(n,r,\log q)$ time a nontrivial $(\pm\lfloor q/2\rfloor)$-zero-sum that uses at most
$$
\frac{1-q^{-1}}{1-q^{-r}}\cdot\ell+r
$$
vectors.
\end{lemma}
\begin{proof}
Similar proof as \Cref{lem:f3-zero-sum-zeros} where we set $q=3$ and $r=\lceil\log_3(\ell+1)\rceil$. 
Let $\ell'\le\ell$ be the rank of the input vectors.
By an invertible linear transform and in $\poly(n,r)$ time, we can assume that we are given standard basis $e_1,\ldots,e_{\ell'}$ and $v_1,\ldots,v_r\in\spn\{e_1,\ldots,e_{\ell'}\}$ as $r$ extra vectors.
Then we consider linear combination $v_\beta=\beta_1v_1+\cdots+\beta_rv_r$ where the coefficient vector $\beta=(\beta_1,\ldots,\beta_r)$ is uniform in $\F_q^r\setminus\{0^r\}$.
A direct calculation shows, in expectation, the number of nonzero entries in $v_\beta$ is at most $\frac{1-q^{-1}}{1-q^{-r}}\cdot\ell'$.
Therefore, once we can find such a $v_\beta$, these nonzero entries can be canceled using $e_1,\ldots,e_{\ell'}$. The total number of used vectors is
$$
\underbrace{r}_\text{for $v_\beta$}+\underbrace{\frac{1-q^{-1}}{1-q^{-r}}\cdot\ell'}_\text{for $e_1,\ldots,e_{\ell'}$}\le\frac{1-q^{-1}}{1-q^{-r}}\cdot\ell+r.
$$
Finally we remark that searching for such a $v_\beta$ can be done efficiently in time $\poly(n,r)$ using the standard derandomization of conditional expectation to sequentially fix $\beta_1,\ldots,\beta_r$.\footnote{Since the vectors are in $\F_q^n$, for each $\beta_i\in\F_q$ there are at most $\min\{q,n+1\}$ choices with different conditional expectation. Hence the enumeration at each derandomization step is $\min\{q,n+1\}=O(n)$, independent of $q$.}
\end{proof}

We note that the $r=1$ case of \Cref{lem:sis_halving_sparse} is the most basic algorithm to find nontrivial linear dependence in $\ell+1$ vectors, as used in \Cref{thm:halving_sis_weak}.
We now incorporate the idea of dimension reduction.

\begin{theorem}\label{thm:sis_halving_start_improve}
Let $q>3$ be a prime and let $r\ge1$ be an arbitrary integer.
Given
$$
m\ge\frac{1-q^{-1}}{1-q^{-r}}\cdot\frac{n(n+1)}2+r(n+1)
$$
vectors in $\F_q^n$, we can find in deterministic $\poly(m,\log q)$ time a nontrivial $(\pm\lfloor q/4\rfloor)$-zero-sum.
\end{theorem}
\begin{proof}
The proof mimics the proof of \Cref{thm:intro_F3-subset-sum}.
Define
$$
m_\ell=\frac{1-q^{-1}}{1-q^{-r}}\cdot\frac{\ell(\ell+1)}2+r(\ell+1).
$$
We prove a more general statement: given $m\ge m_\ell$ vectors in a linear subspace $V$ of $\F_q^n$ with dimension $\ell$, we can efficiently find a $(\pm\lfloor q/4\rfloor)$-zero-sum.
\Cref{thm:sis_halving_start_improve} follows by setting $\ell=n$.

The base case $\ell=0$ is trivial, as $m_0=r\ge1$ and the given vector must be $0$.
For the induction step $\ell\ge1$, since $m_\ell\ge\ell+r$, we apply \Cref{lem:sis_halving_sparse} on the first $\ell+r$ vectors.
We will find a set $T$ of size at most $r+\frac{1-q^{-1}}{1-q^{-r}}\cdot\ell$ and a nontrivial $(\pm\lfloor q/2\rfloor)$-zero-sum from $T$.
Then by \Cref{lem:reducible_sis_halving}, this is turned into a $(\pm\lfloor q/2\rfloor\to\pm\lfloor q/4\rfloor)$-reducible vector $u$ from $T$.

Now if $u=0$ then we are done, since $u$ itself is a nontrivial $(\pm1)$-sum.
Otherwise write $V=\spn\{u\}\oplus W$, where $W$ has dimension $\ell-1$. We will put $T$ aside and work on the remaining $m'=m-|T|\ge m_{\ell-1}$ vectors.
Each remaining vector $v_i$ can be written as $c_iu+w_i$ for some $c_i\in\F_q$ and $w_i\in W$.
Applying induction to the $m'\ge m_{\ell-1}$ many $w_i$'s, we obtain a nontrivial $(\pm\lfloor q/4\rfloor)$-zero-sum of the $w_i$'s.
Now the associated collection of original remaining vectors $v_i$'s sums to $c\cdot u$ for some $c\in\F_q$. Since $u$ is reducible, this can be replaced by a $(\pm\lfloor q/4\rfloor)$-sum from $T$, after which we obtain a nontrivial $(\pm\lfloor q/4\rfloor)$-zero-sum as desired.
\end{proof}

\begin{remark}
A more optimized parameter choice for \Cref{thm:sis_halving_start_improve} is to pick a potentially different $r_i$ for each dimension $i$ and define
$$
m_\ell=\sum_{i=0}^\ell\left(\frac{1-q^{-1}}{1-q^{-r_i}}\cdot i+r_i\right),
$$
where each $r_i\ge1$ is an integer and it automatically guarantees that $m_\ell\ge\ell+r_\ell$.
By setting $r_0=1$ and $r_i=\lceil\log_q(i+1)\rceil$ for $i\ge1$, we have 
$$
m_\ell\le1+\sum_{i=1}^\ell\left(\left(1-q^{-1}\right)(i+1)+\lceil\log_q(i+1)\rceil\right)=\frac{1-q^{-1}}{2}\cdot\ell^2+\ell\log_q(\ell+1)+O(\ell).
$$
However this does not affect the asymptotics in \Cref{thm:simple_halving_sis_intro}.
\end{remark}

With \Cref{thm:sis_halving_start_improve}, we improve \Cref{thm:halving_sis_weak} and  \Cref{thm:simple_halving_sis_intro} follows immediately.

\begin{theorem}\label{thm:halving_sis}
Let $q\ge5$ be a prime and $2\le k \le\lfloor q/2\rfloor$ be an integer power of $2$. Let $r\ge1$ be an arbitrary integer.
Given
$$
m\ge\left(\frac{1-q^{-1}}{1-q^{-r}}\cdot\frac{n(n+1)}2+r(n+1)\right)^{k/2}
$$
vectors in $\F_q^n$, we can find in deterministic $\poly(m,\log q)$ time a nontrivial $(\pm\lfloor q/(2k)\rfloor)$-zero-sum.
\end{theorem}
\begin{proof}
The base case $k=2$ is the algorithm in \Cref{thm:sis_halving_start_improve}.
The inductive case $k\ge4$ follows from \Cref{cor:iterate_sis_halving} with the inductive hypothesis for $k/2$, and noticing $\lfloor\lfloor q/k\rfloor/2\rfloor=\lfloor q/(2k)\rfloor$.
\end{proof}

Different choice of $r$ leads to a different bound in \Cref{thm:halving_sis}. For $r=1$, it generalizes the bound in \Cref{thm:f3-zero-sum-less-weak}. For $r=\lceil\log_q(n+1)\rceil$, it generalizes the bound in \Cref{thm:intro_F3-subset-sum}.
These choices of $r$ readily prove \Cref{thm:simple_halving_sis_intro}.

We also remark that \Cref{thm:halving_sis} (and even its weaker version \Cref{thm:halving_sis_weak}) improves the result of \cite{II24} by a factor of $q^{O(q \log q)}$.

\subsection{\texorpdfstring{$\F_q^n$-\SubsetSum}{Fq-Subset-Sum} on random inputs}\label{sec:avg_subset-sum}

Now we turn to \SubsetSum, a particularly interesting case of the more general $\CIS$ problem.
Imran and Ivanyos \cite{II24} proved a reduction from $(\pm1)$-zero-sum algorithms to subset-zero-sum algorithms that works for worst-case instances but suffers from a blowup of $O(\log q)$ on the exponent.
We identify a more efficient reduction for average-case instances.

We will take a detour to $(0,1,2)$-zero-sum, which means the zero-sum only uses coefficients in $0,1,2$.
For consistency, we use $(0,1)$-zero-sum to denote subset-zero-sum.
Again, we say they are nontrivial if not all coefficients are $0$.

With the help of a worst-case $(\pm1)$-zero-sum algorithm, we show how to combine $(0,1,2)$-zero-sums to obtain a $(0,1)$-zero-sum.

\begin{lemma}\label{lem:subset-sum_012_to_01}
Let $q\ge3$ be a prime.
Suppose there is a deterministic algorithm $\calA$, given $m$ input vectors in $\F_q^n$ and running in time $T$, that outputs a nontrivial $(\pm1)$-zero-sum.

Then given vectors $v_1,\ldots,v_{m'}\in\F_q^n$ and assuming they form $m$ disjoint nontrivial $(0,1,2)$-zero-sums, i.e., for each $i\in[m]$, we have
\begin{equation}\label{eq:lem:subset-sum_012_to_01_1}
\sum_{j\in S_i}\alpha_jv_j=0
\quad\text{where $S_1,\ldots,S_m\subseteq[m']$ are disjoint and nonempty and each $\alpha_j\in\{1,2\}$,}
\end{equation}
we can find in time $T+\poly(m',n)$ a nontrivial $(0,1)$-zero-sum in $v_1,\ldots,v_{m'}$.
\end{lemma}
\begin{proof}
If for some $i\in[m]$ we have $\alpha_j=1$ for all $j\in S_i$, then $\sum_{j\in S_i}v_j=0$ and we are done.
Similarly if for some $i\in[m]$ we have $\alpha_j=2$ for all $j\in S_i$, then $\sum_{j\in S_i}v_j=\sum_{j\in S_i}2\cdot v_j=0$ and we are also done.

Now assume for every $i\in[m]$, there is some $j,j'\in S_i$ satisfying $\alpha_j=1$ and $\alpha_{j'}=2$.
For every $i\in[m]$, rearrange \Cref{eq:lem:subset-sum_012_to_01_1} as
\begin{equation}\label{eq:lem:subset-sum_012_to_01_2}
u_i+2\cdot u_i'=0
\quad\text{where $u_i=\sum_{j\in S_i\colon\alpha_j=1}v_j$ and $u_i'=\sum_{j\in S_i\colon\alpha_j=2}v_j$.}
\end{equation}
By our assumption, each $u_i,u_i'$ is nontrivial.
Now we run $\calA$ on $u_1',\ldots,u_m'$ and obtain a nontrivial $(\pm1)$-zero-sum
$$
\sum_{i\in S}u_i'-\sum_{i\in S'}u_i'=0,
$$
where $S,S'\subseteq[m]$ are disjoint and not both empty.
With \Cref{eq:lem:subset-sum_012_to_01_2}, we obtain the following zero-sum
$$
\sum_{i\in S}u_i'+\sum_{i\in S'}u_i+\sum_{i\in S'}u_i'=\left(\sum_{i\in S}u_i'-\sum_{i\in S'}u_i'\right)+\sum_{i\in S'}(u_i+2\cdot u_i')=0,
$$
which is a desired nontrivial $(0,1)$-zero-sum by expanding each $u_i,u_i'$.
The runtime is obvious from the procedure above.
\end{proof}

A worst-case $(\pm1)$-zero-sum algorithm is provided in \Cref{sec:halving_trick}.
To apply \Cref{eq:lem:subset-sum_012_to_01_1}, it remains to obtain $(0,1,2)$-zero-sums.
Intuitively $0,1,2$ is just a additive shift of $-1,0,1$, which suggests to use $(\pm1)$-zero-sum algorithms to find $(0,1,2)$-zero-sums.
This is formalized as the following lemma.

\begin{lemma}\label{lem:subset-sum_012_from_pm1}
Let $q\ge3$ be a prime. Suppose there is a deterministic algorithm $\calA$, given $m$ input vectors in $\F_q^n$ and running in time $T$, that outputs a nontrivial $(\pm1)$-zero-sum.

Let $\eps\in(0,1]$ be arbitrary and define $d=\lceil\log(q/\eps)\rceil$.
Then there is a deterministic algorithm $\calB$, given $dm+1$ input vectors in $\F_q^n$ and running in time $dT+\poly(m,q,\log(1/\eps))$, that outputs a nontrivial $(0,1,2)$-zero-sum with probability at least $1-\eps$ over uniform random inputs.
\end{lemma}
\begin{proof}
Let $v_1,\ldots,v_{dm+1}$ be the input vectors and define $v^*=-\sum_{i\in[dm+1]}v_i$.
If $v^*=0$, then we are done, as $\sum_{i\in[dm+1]}v_i=-v^*=0$ is a trivial $(0,1,2)$-zero-sum.
Now we assume $v^*\ne0$.
If we can find a $(\pm1)$-sum of $v_1,\ldots,v_{dm}$ that equals $v^*$, then by expanding $v^*$ and rearranging, we obtain a $(0,1,2)$-zero-sum. Note that such a $(0,1,2)$-zero-sum is nontrivial since $v_{dm+1}$ always has coefficient $1$.

To find the target $(\pm1)$-sum, we write $\F_q^n=\spn\{v^*\}\oplus W$ where $W$ has dimension $n-1$.
Let $i^*\in[n]$ be a nonzero coordinate of $v^*$, i.e., $v^*[i^*]\ne0$.
For each $i\in[dm]$, we express $v_i$ as $c_iv^*+w_i$ where $c_i=v_i[i^*]/v^*[i^*]\in\F_q$ and $w_i=v_i-c_iv^*\in W$.
Since $v_1,\ldots,v_{dm+1}$ are independent and uniform, we know $c_1,\ldots,c_{dm}$ are independent and uniform even when we condition on $v^*$.

Now we divide $v_1,\ldots,v_{dm}$ into $d$ batches $S_1,\ldots,S_d$ of $m$ vectors. 
On each batch $j\in[d]$, we run $\calA$ on $w_i,i\in S_j$ to obtain a nontrivial $(\pm1)$-zero-sum in time $T$:
$$
\sum_{i\in S_j}\alpha_iw_i=0
\quad\text{where $\alpha_i\in\{0,\pm1\}$ are not all zero.}
$$
This means $\sum_{i\in S_j}\alpha_iv_i=\beta_jv^*$ where $\beta_j=\sum_{i\in S_j}\alpha_ic_i$.
Suppose there is some set $T\subseteq[d]$ that $\sum_{j\in T}\beta_j=1$. We can use the textbook dynamic programming algorithm to find such a set in time $\poly(q,d)$, then obtain a desired $(\pm1)$-sum
$$
\sum_{j\in T}\sum_{i\in S_j}\alpha_iw_i=v^*.
$$

It remains to show such a $T$ exists with probability $1-\eps$.
Since $\alpha_i,i\in S_j$ are not all zero and the $c_i$'s are independent and uniform, $\beta_1,\ldots,\beta_d\in\F_q$ are independent and uniform.
For each $\emptyset\ne T\subseteq[d]$, define $X_T\in\{0,1\}$ as the indicator that $\sum_{j\in T}\beta_j=1$.
Let $X=\sum_{\emptyset\ne T\subseteq[d]}X_T$.
Then we have $\E[X]=(2^d-1)/q$ and $\Var[X]=(2^d-1)(q^{-1}-q^{-2})$, since $\E[X_T]=1/q$ and the $X_T$'s are pairwise independent.
In addition, since 
$$
\Var[X]=\sum_{k\ge0}\Pr[X=k]\cdot(k-\E[X])^2\ge\Pr[X=0]\cdot\E[X]^2,
$$
we have 
$$
\Pr[\text{such a $T$ does not exist}]
=\Pr[X=0]\le\frac{\Var[X]}{\E[X]^2}
=\frac{q-1}{2^d-1}\le\frac q{2^d}\le\eps.
$$
This completes the proof.
\end{proof}

At this point, we are ready to present our algorithm that finds a subset-zero-sum in random inputs.

\begin{theorem}\label{thm:subset-sum_avg}
Let $q\ge5$ be a prime and define $k$ the unique integer power of $2$ with $q/4<k\le\lfloor q/2\rfloor$.
Let $\eps\in(0,1]$ be arbitrary and let $1\le r\le n$ be an arbitrary integer.
Given 
$$
m\ge\left(\frac{1-q^{-1}}{1-q^{-r}}\cdot\frac{n(n+1)}2+r(n+1)\right)^k\cdot O\left(\log(q)\log(1/\eps)\right)
$$
uniform random vectors in $\F_q^n$, with probability at least $1-\eps$ we can find in deterministic $\poly(m,q)$ time a nontrivial $(0,1)$-zero-sum.
\end{theorem}
\begin{proof}
For any real number $x\ge1$, there is a unique integer power of $2$ that is greater than $x$ and at most $\lfloor 2x\rfloor$.
This shows the existence and uniqueness of our choice of $k$, which also guarantees $\lfloor q/(2k)\rfloor=1$.
Let 
$$
\bar m=\left(\frac{1-q^{-1}}{1-q^{-r}}\cdot\frac{n(n+1)}2+r(n+1)\right)^{k/2}
$$
be the bound in \Cref{thm:halving_sis} such that we can always efficiently find a nontrivial $(\pm1)$-zero-sum in $\bar m$ input vectors.

Let $d=O\left(\log q\right)$.
Then by \Cref{lem:subset-sum_012_from_pm1}, with probability $0.99$, we can find a nontrivial $(0,1,2)$-zero-sum in $d\cdot\bar m+1$ uniform random input vectors.
Hence by standard concentration, with probability at least $1-\eps$ we can find $\bar m$ disjoint nontrivial $(0,1,2)$-zero-sums if we run the above procedure on $O(\bar m\log(1/\eps))$ many independent batches of $d\cdot\bar m+1$ uniform random input vectors.
Finally we apply \Cref{lem:subset-sum_012_to_01} on these $(0,1,2)$-zero-sums to obtain a nontrivial $(0,1)$-zero-sum.

The total number of input vectors is $O(\bar m\log(1/\eps))\cdot(d\cdot\bar m+1)$ as claimed.
\end{proof}

The reduction in \cite{II24} produces a bound of $m\ge q^{O(q\log^2q)}n^{O(q\log q)}$, whereas our \Cref{thm:subset-sum_avg} produces $m\ge n^{O(q)}$.

\subsection{Simple \texorpdfstring{$\CIS$}{CIS} algorithms: large primes}\label{sec:simple_cis}

Finally we use the simple tools developed so far to give simple algorithms for the $\CIS$ problem and prove \Cref{thm:simple_cis_lev_intro}.
This will already dequantize, and in fact greatly improve, \Cref{thm:CLZCIS}, when the field size $q$ grows.

Recall that the $\CIS$ problem will only allow coefficients from some allowed set $A\subseteq\F_q$.
For consistency with previous sections, we use $A$-sum to denote a vector sum whose coefficients are all from $A$. We say it is $A$-zero-sum if the vector sum also equals $0$.
Then a solution to the $\CIS$ problem is equivalent to a nontrivial $A$-zero-sum.

For any $a\in\F_q$, we use $aA$ to denote the set $\{aa'\colon a'\in A\}$ and use $A+a$ to denote the set $\{a+a'\colon a'\in A\}$.

The following reduction is a natural generalization of \Cref{lem:subset-sum_012_from_pm1}. In particular, \Cref{lem:subset-sum_012_from_pm1} is equivalent to setting $A=\{0,\pm1\},a=1,b=1$ in \Cref{lem:cis_worst-avg}.

\begin{lemma}\label{lem:cis_worst-avg}
Let $q\ge3$ be a prime and let $A\subseteq\F_q$ be nonempty. Suppose there is a deterministic algorithm $\calA$, given $m$ input vectors in $\F_q^n$ and running in time $T$, that outputs a nontrivial $A$-zero-sum.\footnote{Since $\calA$ works in the worst case and outputs nontrivial $A$-zero-sum, this immediately forces $0\in A$. Hence we do not need to explicitly assume $0\in A$ in the above statement.}

Let $\eps\in(0,1]$ be arbitrary and define $d=\lceil\log(q/\eps)\rceil$.
Fix arbitrary $a\in\F_q\setminus\{0\},b\in\F_q$ and define $B=aA+b$.
Then there is a deterministic algorithm $\calB$, given $dm+1$ input vectors in $\F_q^n$ and running in time $dT+\poly(m,q,\log(1/\eps))$, that outputs a nontrivial $B$-zero-sum with probability at least $1-\eps$ over uniform random inputs.
\end{lemma}
\begin{proof}
We safely assume $a=1$ since any nontrivial zero-sum with coefficients in $A+(b/a)$ can be dilated, by multiplying every coefficient with $a$, as a nontrivial zero-sum with coefficients in $a(A+(b/a))=B$.
If $b=0$, then we just run $\calA$ on $m$ input vectors to obtain a nontrivial $A$-zero-sum.
Now we assume $a=1$ and $b\ne0$.

Let $v_1,\ldots,v_{dm+1}$ be input vectors. We define $v^*=-b\sum_{i\in[dm+1]}v_i$.
It suffices to find an $A$-sum of $v_1,\ldots,v_{dm}$ that equals $v^*$, then by expanding $v^*$ and rearranging, we obtain a $B$-zero-sum.
Note that such a $B$-zero-sum is nontrivial since $v_{dm+1}$ has coefficient $b\ne0$.
The rest of the analysis is identical to the proof of \Cref{lem:subset-sum_012_from_pm1}. 
\end{proof}

Combining \Cref{lem:cis_worst-avg} with the $\SIS$ algorithms in \Cref{sec:halving_trick}, we can handle any allowed coefficient set with long arithmetic progression.

\begin{definition}[Arithmetic Progression (AP)]\label{def:AP}
Let $A\subseteq\F_q$ have size $c\ge2$. We say $A$ is a $t$-AP (equivalently, an AP of length $t$) if there are $x\in\F_q$ and $0\ne y\in\F_q$ such that $A=\{x,x+y,x+2y,\ldots,x+(t-1)y\}$.
We say $A$ contains a $t$-AP if some $A'\subseteq A$ is a $t$-AP.
\end{definition}

\begin{corollary}\label{cor:simple_cis_general}
Let $q\ge5$ be a prime and let $A\subseteq\F_q$ be nonempty.
Assume $A$ contains an AP of length $1+2\lfloor q/(2k)\rfloor$ where $k\ge2$ that is an integer power of $2$.
Let $r\ge1$ be an arbitrary integer.
Given
$$
m\ge1+\left(\frac{1-q^{-1}}{1-q^{-r}}\cdot\frac{n(n+1)}2+r(n+1)\right)^{k/2}\cdot\lceil\log(q/\eps)\rceil
$$
uniform random vectors in $\F_q^n$, with probability at least $1-\eps$ we can find in deterministic $\poly(m,q)$ time a nontrivial $A$-zero-sum.
\end{corollary}
\begin{proof}
Firstly, a direct brute-force algorithm finds the AP in $A$ in time $\poly(q)$.
Then we can assume without loss of generality that $A$ is the AP.
By \Cref{lem:cis_worst-avg}, it suffices to have a worst-case algorithm for nontrivial $(\pm\lfloor q/(2k)\rfloor)$-zero-sums, for which we use \Cref{thm:halving_sis}.
\end{proof}

In the case of large $q$, we have known results in arithmetic combinatorics that guarantee long APs. The following fact is due to Lev \cite{Lev00} and, in \Cref{app:arith_comb}, we provide a short proof for completeness.

\begin{restatable}[\cite{Lev00}]{fact}{fctlevlongAP}\label{fct:lev_long_AP}
Let $q\ge5$ be a prime.
Let $A\subseteq\F_q$ be an arbitrary set of size $|A|\ge q-\log_4(q+2)$.
Then $A$ contains an AP of length $(q+1)/2$.
\end{restatable}

Putting together \Cref{fct:lev_long_AP} and \Cref{cor:simple_cis_general}, we obtain our simple $\CIS$ algorithms for large primes.

\begin{theorem}\label{thm:simple_cis_lev}
Let $q\ge5$ be a prime and let $A\subseteq\F_q$ be an arbitrary set of size $|A|\ge q-\log_4(q+2)$.
Let $r\ge1$ be an arbitrary integer.
Given
$$
m\ge1+\left(\frac{1-q^{-1}}{1-q^{-r}}\cdot\frac{n(n+1)}2+r(n+1)\right)\cdot\lceil\log(q/\eps)\rceil
$$
uniform random vectors in $\F_q^n$, with probability at least $1-\eps$ we can find in deterministic $\poly(m,q)$ time a nontrivial $A$-zero-sum.
\end{theorem}
\begin{proof}
By \Cref{fct:lev_long_AP}, $A$ contains an AP of length $(q+1)/2$.
Since $(q+1)/2\ge1+2\lfloor q/4\rfloor$, we set $k=2$ and apply \Cref{cor:simple_cis_general}.
\end{proof}

Finally we prove \Cref{thm:simple_cis_lev_intro}, which is a simple deduction from \Cref{thm:simple_cis_lev}.
\begin{proof}[Proof of \Cref{thm:simple_cis_lev_intro}]
Since $q>4^{k-1}$ and $k\ge2$, we know $q\ge5$ and the number of allowed coefficients is $q-k+1\ge q-\log_4(q+2)$.
Hence we set $r=1$ and $\eps\in(0,1]$ to be a small constant then apply \Cref{thm:simple_cis_lev}. This proves \Cref{thm:simple_cis_lev_intro}.
\end{proof}


\section{Beyond the halving trick}\label{sec:general_cis}

In this section, we expand our toolbox by developing more general reductions.

To give some motivation, consider the setting where we want to find a nontrivial $(\pm\lfloor q/6\rfloor)$-zero-sum.
So far we only have the halving trick, and thus we will pay a \emph{quartic} blowup by halving twice, which ends up providing a nontrivial $(\pm\lfloor q/8\rfloor)$-zero-sum and is more than we need.
Our new reductions in this section will allow us to pay a \emph{cubic} cost to find nontrivial $(\pm\lfloor q/6\rfloor)$-zero-sums.

\paragraph{Notation.}
We will generalize existing notions as follows.
For $A,B\subseteq\F_q$, we define $A-B=\{a-b\colon a\in A,b\in B\}$. We use $\pm A$ to denote $A\cup(-A)$.
We say sets $A_1,\ldots,A_k$ is a nonempty disjoint partition of $A$ if $A_1,\ldots,A_k$ are all nonempty and pairwise disjoint and $A=A_1\cup\cdots\cup A_k$.

\begin{definition}[General Vector Sum]\label{def:vector_sum_general}
Let $v_1,\ldots,v_m\in\F_q^n$ be vectors and let $H\subseteq\F_q$ be nonempty.
We say vector $u$ is an $H$-sum of the $v_i$'s if we can choose $\alpha_i\in H$ for each $i\in[m]$ such that $u=\sum_{i\in[m]}\alpha_iv_i$; and it is an $H$-zero-sum if $u=0$. 

We say $u$ is $s$-sparse if it has at most $s$ nonzero coefficients.
Let $\emptyset\ne A\subseteq\F_q\setminus\{0\}$.
We say $u$ is $A$-nontrivial if for any $a\in A$ there exists some $i\in[m]$ such that $\alpha_i\in\{a,-a\}$.
\end{definition}

In comparison with \Cref{def:vector-sum_sis}, a $(\pm h)$-sum is equivalent to an $H$-sum where $H=\{-h,-h+1,\ldots,h\}$ and a nontrivial sum is equivalent to an $A$-nontrivial sum for some $\emptyset\ne A\subseteq\F_q\setminus\{0\}$.
In particular, a $(\pm\lfloor q/2\rfloor)$-zero-sum is an $\F_q$-zero-sum.

\begin{definition}[General Reducible Vector]\label{def:reducible_vector_general}
Let $v_1,\ldots,v_m\in\F_q^n$ be vectors and let $H,H'\subseteq\F_q$ and let $\emptyset\ne A\subseteq\F_q\setminus\{0\}$.
We say that vector $u$ is $(H\to H')$-reducible for the $v_i$'s if
\begin{itemize}
\item $u$ is a nontrivial $(\pm1)$-sum of the $v_i$'s (i.e., $u$ is a $\{1\}$-nontrivial $\{0,\pm1\}$-sum of the $v_i$'s);
\item for every $c\in H$, the vector $c\cdot u$ is an $H'$-sum of the $v_i$'s.
\end{itemize}

We say $u$ is $s$-sparse if there exists a set $S\subseteq[m]$ of size $|S|\le s$ such that $u$ is $(H\to H')$-reducible for $v_i,i\in S$.
Let $\emptyset\ne A\subseteq\F_q\setminus\{0\}$. We say $u$ is $A$-nontrivial if for every $c\in H$, $c\cdot u$ is an $A$-nontrivial $H'$-sum.
\end{definition}

In comparison with \Cref{def:reducible_sis}, an $(\pm h\to\pm h')$-reducible vector is a $(H\to H')$-reducible vector where $H=\{-h,-h+1,\ldots,h\}$ and $H'=\{-h',-h'+1,\ldots,h'\}$.

Whenever we construct a reducible vector $u$, we will provide a procedure to support efficient query: given $c\in H$, we can express $c\cdot u$ as an $H'$-sum (or $A$-nontrivial $H'$-sum) of the $v_i$'s in deterministic time $\poly(m,n,\log q)$. Hence to ease presentation, we will omit this detail and only refer to the construction of reducible vectors.

\subsection{Basic zero-sum algorithms for specific nontriviality}\label{sec:general_cis_basic_zero-sum}

Given \Cref{def:vector_sum_general}, it is natural to ask for basic algorithms that ensure the specific $A$-nontriviality.
It turns out that it is not too hard to modify the existing \Cref{lem:sis_halving_sparse} to obtain what we need.

\begin{lemma}\label{lem:general_cis_starting_point}
Let $\emptyset\ne A\subseteq\F_q\setminus\{0\}$ and let $r\ge1$ be an arbitrary integer.
Given $n+r|A|$ vectors in $\F_q^n$, we can find in deterministic $\poly(n,r,|A|,\log q)$ time an $s$-sparse $A$-nontrivial $\F_q$-zero-sum, where
$$
s=\frac{1-q^{-1}}{1-q^{-r}}\cdot n+r|A|.
$$
\end{lemma}
\begin{proof}
Denote $A=\{a_1,\ldots,a_c\}$ where $c=|A|$.
Let $\ell\le n$ be the rank of the input vectors. By an invertible linear transform and in $\poly(n,r,|A|)$ time, we can assume that we are given standard basis $e_1,\ldots,e_\ell$ and $v_1^{(i)},\ldots,v_c^{(i)},i\in[r]$ as $rc$ extra vectors that lie in the span of $e_1,\ldots,e_\ell$.

For each $i\in[r]$ we compute
\begin{equation}\label{eq:lem:general_cis_starting_point_1}
u^{(i)}=a_1v_1^{(i)}+\cdots+a_cv_c^{(i)}.
\end{equation}
Then by the same proof as \Cref{lem:sis_halving_sparse}, we can find $\beta_1,\ldots,\beta_r\in\F_q$ not all zero such that
\begin{equation}\label{eq:lem:general_cis_starting_point_2}
u_\beta=\beta_1u^{(1)}+\cdots+\beta_ru^{(r)}
\end{equation}
has at most $\frac{1-q^{-1}}{1-q^{-r}}\cdot\ell$ nonzero entries.

Since $\beta_1,\ldots,\beta_r$ are not all zero, we can safely assume $\beta_j=1$ for some $j\in[r]$.
Hence by substituting \Cref{eq:lem:general_cis_starting_point_1} into \Cref{eq:lem:general_cis_starting_point_2}, we know that $a_1,\ldots,a_c$ all appear as coefficients of $u_\beta$, which implies that $u_\beta$ is $A$-nontrivial.

Finally we use the standard basis $e_1,\ldots,e_\ell$ to cancel the nonzero entries of $u_\beta$, which provides the desired zero-sum.
\end{proof}

\subsection{From zero-sum to reducible vector}\label{sec:general_cis_zero-sum_to_reducible}

Similar to \Cref{lem:reducible_sis_halving}, we can obtain $(H\to H')$-reducible vectors from $H$-zero-sums.
However things are more complicated in the general setting.

We start with the simplest case that we partition $H$ into two parts and one of which becomes (part of) $H'$.
It is encouraged to keep in mind the example of \Cref{lem:reducible_sis_halving} where $H=\{-h,\ldots,h\}$ and $H'=\{-\lfloor h/2\rfloor,\ldots,\lfloor h/2\rfloor\}$.

\begin{lemma}\label{lem:general_cis_larger_k-simple}
Let $H\subseteq\F_q$ and let $\emptyset\ne A\subseteq\F_q\setminus\{0\}$.
Suppose there is a deterministic algorithm $\calA$, given $m$ vectors in $\F_q^n$ and running in time $T$, that outputs an $s$-sparse $A$-nontrivial $H$-zero-sum.

Let $H_0,H_1$ be a nonempty disjoint partition of $H$.
Assume $A_0:=(\pm A)\cap H_0$ and $A_1:=(\pm A)\cap H_1$ are nonempty.
Let $B_1\subseteq\F_q$ be arbitrary such that $H_1\subseteq\pm B_1$.
Define
$$
H'=(\pm H_0)\cup(B_1-B_1).
$$
Then there is a deterministic algorithm $\calB$, given $m$ vectors in $\F_q^n$ and running in time $T+\poly(m,n,\log q)$, that outputs an $s$-sparse $A_0$-nontrivial $(H_1\to H')$-reducible vector.
\end{lemma}
\begin{proof}
Let $v_1,\ldots,v_m\in\F_q^n$ be the input vectors. We apply $\calA$ to obtain an $s$-sparse $A$-nontrivial $H$-zero-sum $\alpha_1v_1+\cdots+\alpha_mv_m$.
Let $S\subseteq[m]$ of size $|S|\le s$ be the nonzero coefficients.
Compute $S_0=\{j\in S\colon a_j\in H_0\}$ and $S_1=\{j\in S\colon a_j\in H_1\}$, which is a disjoint partition of $S$.
Moreover, since it is $A$-nontrivial and $A_0\subseteq H_0,A_1\subseteq H_1$ are nonempty, we know $S_0,S_1$ are nonempty.

For each $j\in S_1$, we define $\beta_j\in\{\pm1\}$ such that $\alpha_j\beta_j\in B_1$, i.e., $\beta_j=1$ if $\alpha_j\in B_1$; and $\beta_j=-1$ if $-\alpha_j\in B_1$.
Now we can rewrite the $A$-nontrivial $B$-zero-sum as
\begin{equation}\label{eq:lem:general_cis_larger_k-simple_1}
\sum_{j\in S_0}\alpha_jv_j+\sum_{j\in S_1}\alpha_j\beta_j\cdot\beta_jv_j=0.
\end{equation}
By our construction above, $\sum_{j\in S_0}\alpha_jv_j$ is an $A_0$-nontrivial $H_0$-sum and $\sum_{j\in S_1}\alpha_j\beta_j\cdot\beta_jv_j=\sum_{j\in S_1}\alpha_jv_j$ is an $A_1$-nontrivial $H_1$-sum.

At this point, we define
\begin{equation}\label{eq:lem:general_cis_larger_k-simple_2}
u=\sum_{j\in S_1}\beta_jv_j.
\end{equation}
We will prove that $u$ is a $A_0$-nontrivial $(H_1\to H')$-reducible vector.

As mentioned above, $\sum_{j\in S_1}\alpha_j\beta_j\cdot\beta_jv_j$ is an $A_1$-nontrivial $H_1$-sum; hence $u$ is a nontrivial $(\pm1)$-sum of $v_j,j\in S$.
Let $c\in H_1$ and, by symmetry, assume $c\in B_1$.
Then we have
\begin{align*}
c\cdot u
&=\sum_{j\in S_1}c\beta_jv_j-\left(\sum_{j\in S_0}\alpha_jv_j+\sum_{j\in S_1}\alpha_j\beta_j\cdot \beta_jv_j\right)
\tag{by \Cref{eq:lem:general_cis_larger_k-simple_1}}\\
&=\underbrace{-\sum_{j\in S_0}\alpha_jv_j}_P+\underbrace{\sum_{j\in S_1}(c-\alpha_j\beta_j)\cdot\beta_jv_j}_Q.
\tag{by \Cref{eq:lem:general_cis_larger_k-simple_2}}
\end{align*}
Note $P$ has coefficients in $\pm H_0\subseteq H'$ and is $A_0$-nontrivial; and $Q$ has coefficients in $B_1-B_1\subseteq H'$. 
Hence $u=P+Q$ is an $A$-nontrivial $(H\to H')$-reducible vector.
Since we only use vector in $S$, it is also $s$-sparse.
The runtime follows directly from the description above.
\end{proof}

In the following \Cref{lem:general_cis_larger_k}, we proceed to the general case, where we partition $H$ into more $k+1\ge3$ parts.
Analogously, it is helpful to keep in mind the example where $H=\{-h,\ldots,h\}$, $H_0=\{-\lfloor h/(k+1)\rfloor,\ldots,\lfloor h/(k+1)\rfloor\}$, and each $H_i,i\in[k]$ is an interval of length at most $\lceil h/(k+1)\rceil$.

We will need the following standard notions on permutations.
Let $k\ge1$ be an integer. We use $\calS_k\subsetneq[k]^k$ to denote the symmetric group on $[k]$ elements. Each $\pi=(\pi_1,\ldots,\pi_k)\in\calS_k$ defines the action $\pi(i)=\pi_i$ for all $i\in[k]$. We also define $\pi^{-1}(i)\in[k]$ to be the location of $i$ in $\pi$, i.e., $\pi_{\pi^{-1}(i)}=i$.
We use $\sgn(\pi)\in\{\pm1\}$ to denote the sign of the permutation $\pi$.

\begin{lemma}\label{lem:general_cis_larger_k}
Let $H\subseteq\F_q$ and let $\emptyset\ne A\subseteq\F_q\setminus\{0\}$.
Suppose there is a deterministic algorithm $\calA$, given $m(n)$ vectors in $\F_q^n$ and running in time $T(n)$, that outputs an $s(n)$-sparse $A$-nontrivial $H$-zero-sum.

Let $H_0,H_1,\ldots,H_k$ be a nonempty disjoint partition of $H$ where $k\ge2$.
Assume $A_i:=(\pm A)\cap H_i$ is nonempty for all $i=0,\ldots,k$.
For $i\in[k]$, let $B_i\subseteq\F_q$ be arbitrary such that $H_i\subseteq\pm B_i$.
Define
$$
H'=(\pm H_0)\cup\bigcup_{i\in[k]}(B_i-B_i).
$$
Define $m_1=1$ and, for $\ell=2,3,\ldots,k+1$,
\begin{align*}
m_\ell
&=(-1+m_{\ell-1})\cdot s(k^{k-\ell+1}n)+m(k^{k-\ell+1}n)\\
&=\prod_{j=1}^{\ell-1}s(k^{k-j}n)+\sum_{i=1}^{\ell-1}\left(m(k^{k-i}n)-s(k^{k-i}n)\right)\prod_{j=1}^{i-1}s(k^{k-j}n).
\end{align*}
Then there is a deterministic algorithm $\calB$, given $m_{k+1}$ vectors in $\F_q^n$ and running in time $T'$, that outputs an $s'$-sparse $A_0$-nontrivial $(H_1\cup\cdots\cup H_k\to H')$-reducible vector, where
$$
s'=\prod_{\ell=1}^ks(k^{k-\ell}n)
\qquad\text{and}\qquad
T'=\sum_{\ell=1}^km_\ell\cdot T(k^{k-\ell}n)+\poly(m_{k+1},n,k,\log q).
$$
\end{lemma}
\begin{proof}
Let $v_1,\ldots,v_{m_{k+1}}$ be the input vectors.
We will build a forest of rooted trees of depth at most $k+1$ by iteratively applying $\calA$.
\begin{enumerate} 
\item\label{itm:lem:general_cis_larger_k_depth}
For each depth $\ell=1,2,\ldots,k+1$, there are exactly $m_\ell$ nodes.
The leaves are at depth $k+1$.
\item\label{itm:lem:general_cis_larger_k_leaf}
The $j$th leaf is simply denoted $j$.
We define $Q^j=v_j\in\F_q^n$ and denote $Q^j[\emptyset]=v_j$.
\item\label{itm:lem:general_cis_larger_k_node}
For depth $\ell=k,k-1,\ldots,1$, we first list all $m_{\ell+1}$ nodes, denoted $x_1,\ldots,x_{m_{\ell+1}}$, at depth $\ell+1$.
Then we create $m_\ell$ nodes at depth $\ell$ as follows.

Each time we create a depth-$\ell$ node $z$, we pick an arbitrary set $S\subseteq\{x_1,\ldots,x_{m_{\ell+1}}\}$ of the remaining depth-$(\ell+1)$ nodes of size $|S|=m(k^{k-\ell}n)$.
Each $x\in S$ is already associated with some vector $Q^x\in\F_q^{k^{k-\ell}n}$; so we apply $\calA$ on $Q^x,x\in S$ to obtain an $s(k^{k-\ell}n)$-sparse $A$-nontrivial $H$-zero-sum
$$
\sum_{x\in T}\alpha_xQ^x=0,
$$
where $T\subseteq S$ contains the nonzero coefficients and has size $|T|\le s(k^{k-\ell}n)$.
Then we connect $x\in T$ as the child nodes of $z$ and discard them from the remaining depth-$(\ell+1)$ nodes.

Since each time we discard at most $s(k^{k-\ell}n)$ nodes, we are guaranteed with at least $m(k^{k-\ell}n)$ remaining ones for creating each one of the $m_\ell$ many $z$'s.
Runtime here is $m_\ell\cdot T(k^{k-\ell}n)$ for executions of $\calA$.

Finally we record the following quantities.
\begin{enumerate}
\item\label{itm:lem:general_cis_larger_k_node_1} 
Nonzero coefficients $\alpha_x$ of the zero-sum above for each $x\in T$.
\item\label{itm:lem:general_cis_larger_k_node_2} 
Nonempty sets $S^z_i\subseteq T$ for each $i=0,\ldots,k$ corresponding to coefficients in $H_i$, such that $|S_0^z|+\cdots+|S_k^z|=|T|\le s(k^{k-\ell}n)$ and $S_i^z\cap S_{i'}^z=\emptyset$ whenever $i\ne i'$.

The nonemptiness and disjointness are because the $H$-zero-sum is $A$-nontrivial, where $H_0,\ldots,H_k$ is a nonempty disjoint partition of $H$ with nonempty $A_0\subseteq H_0,\ldots,A_k\subseteq H_k$.
\item\label{itm:lem:general_cis_larger_k_node_3} 
Signs $\beta_x\in\{\pm1\}$ for each $x\in T$ such that $\alpha_x\beta_x\in B_i$ where $i=0,1,\ldots,k$ is the unique index that $x\in S_i^z$.
\item\label{itm:lem:general_cis_larger_k_node_4}
$A_i$-nontrivial $H_i$-sum $w^z_i=\sum_{x\in S_i^z}\alpha_xQ^x$ for each $i=0,\ldots,k$, such that $w_0^z+\cdots+w_k^z=0$.
\item\label{itm:lem:general_cis_larger_k_node_5}
Nontrivial $(\pm1)$-sum $Q_i^z=\sum_{x\in S_i^z}\beta_xQ^x$ for each $i\in[k]$.

Note that $Q_i^z$ is nontrivial as $S_i^z\ne\emptyset$.
\item\label{itm:lem:general_cis_larger_k_node_6}
Vector $Q^z\in\F_q^{k^{k-\ell+1}n}$ defined by tupling $Q^z_1,\ldots,Q^z_k$.

Alternatively, we view $Q^z\in(\F_q^n)^{k\times\cdots\times k}$ as a $(k-\ell+1)$-dimensional array.
For $R=(r_1,\ldots,r_{k-\ell+1})\in[k]^{k-\ell+1}$, we recursively define the $R$th entry of $Q^z$, denoted $Q^z[R]\in\F_q^n$, as 
$$
Q^z[R]=Q^z_{r_1}[r_2,\ldots,r_{k-\ell+1}]=\sum_{x\in S_{r_1}^z}\beta_xQ^x[r_2,\ldots,r_{k-\ell+1}].
$$
\end{enumerate}
\end{enumerate}
The total runtime $T'$ consists of $\sum_{\ell=1}^km_\ell\cdot T(k^{k-\ell}n)$ for executions of $\calA$ and $\poly(m_{k+1},n,k,\log q)$ for other processing time.

Note that there is a single node $z_0$ at depth $1$ in the above construction. We use $\calT$ to denote the depth-$(k+1)$ tree rooted at $z_0$.
For our reducible vector $u$, we will only use input vectors that are leaves of $\calT$. Since each depth-$\ell$ node has at most $s(k^{k-\ell}n)$ child nodes, $\calT$ has at most $\prod_{\ell=1}^ks(k^{k-\ell}n)$ leaves. This proves the sparsity bound $s'$.

Now we show how to obtain a reducible vector $u$ given $\calT$.
Recall that $\calS_k\subsetneq[k]^k$ is the set of permutations on $[k]$.
Define
\begin{align}
u
&=\sum_{\pi\in\calS_k}\sgn(\pi)Q^{z_0}[\pi]
\notag\\
&=\sum_{\pi=(\pi_1,\ldots,\pi_k)\in\calS_k}\sgn(\pi)\sum_{z_1\in S^{z_0}_{\pi_1}}\beta_{z_1}\sum_{z_2\in S^{z_1}_{\pi_2}}\beta_{z_2}\cdots\sum_{z_k\in S_{\pi_k}^{z_{k-1}}}\beta_{z_k}Q^{z_k}[\emptyset]
\tag{by \Cref{itm:lem:general_cis_larger_k_node_6}}\\
&=\sum_{\pi=(\pi_1,\ldots,\pi_k)\in\calS_k}\sgn(\pi)\sum_{z_1\in S^{z_0}_{\pi_1}}\beta_{z_1}\sum_{z_2\in S^{z_1}_{\pi_2}}\beta_{z_2}\cdots\sum_{z_k\in S_{\pi_k}^{z_{k-1}}}\beta_{z_k}v_{z_k}
\tag{by \Cref{itm:lem:general_cis_larger_k_leaf}}\\
&=\sum_{\pi=(\pi_1,\ldots,\pi_k)\in\calS_k}\sgn(\pi)\sum_{z_1\in S^{z_0}_{\pi_1},\ldots,z_k\in S_{\pi_k}^{z_{k-1}}}\left(\prod_{\ell\in[k]}\beta_{z_\ell}\right)v_{z_k}.
\label{eq:lem:general_cis_larger_k_2}
\end{align}
The intuition of $u$ is as follows. Each layer in the tree enforces a decomposition of an $H$-zero-sum into components aligned with $H_0, H_1, \ldots, H_k$; the permutation sum with alternating signs implements a determinant-like cancellation across branches so that for any $c \in H_i$, the combination $c \cdot u$ decomposes into $H'$-conforming pieces (e.g., $\pm H_0$ and $B_i - B_i$) while preserving $A$-nontriviality.

By \Cref{itm:lem:general_cis_larger_k_node_2} and \Cref{itm:lem:general_cis_larger_k_node_5}, $u$ is a nontrivial $(\pm1)$-sum.
To prove that $u$ is an $A_0$-nontrivial $(H_1\cup\cdots\cup H_k\to H')$-reducible vector, let $i^*\in[k]$ and $c\in H_{i^*}$ be arbitrary. By symmetry, we assume $c\in B_{i^*}$ and analyze $c\cdot u$.

Fix an arbitrary $\pi=(\pi_1,\ldots,\pi_k)\in\calS_k$.
Let $j=\pi^{-1}(i^*)$.
Define $\tilde\pi=(\pi_1,\ldots,\pi_{j-1})$ and $\bar\pi=(\pi_{j+1},\ldots,\pi_k)$; so $\pi=(\tilde\pi,i^*,\bar\pi)$.
We use $\calT[\tilde\pi]$ to denote the set of nodes $z$ can be reached by the following process: start with $z=z_0$; for each step $d=1,\ldots,j-1$, proceed to an arbitrary child node of $z$ in $S_{\pi_d}^z$.
We have the following observation.

\begin{claim}\label{clm:lem:general_cis_larger_k_2}
$\calT[\tilde\pi]$ is a nonempty set of depth-$j$ nodes.
Let $z_{j-1}\in\calT[\tilde\pi]$ be arbitrary. 
Then
$$
\sum_{i=0}^k\sum_{z_j\in S_i^{z_{j-1}}}\alpha_{z_j}\beta_{z_j}\sum_{z_{j+1}\in S_{\pi_{j+1}}^{z_j},\ldots,z_k\in S_{\pi_k}^{z_{k-1}}}\left(\prod_{\ell=j}^k\beta_{z_\ell}\right)v_{z_k}=0.
$$
\end{claim}
\begin{proof}
The nonemptiness and depth condition follow directly from \Cref{itm:lem:general_cis_larger_k_node_2}.
By \Cref{itm:lem:general_cis_larger_k_node_4}, we have $\sum_{i=0}^k\sum_{z_j\in S_i^{z_{j-1}}}\alpha_{z_j}Q^{z_j}=w_0^{z_{j-1}}+\cdots+w_k^{z_{j-1}}=0$.
By \Cref{itm:lem:general_cis_larger_k_node_6}, each $Q^{z_j}$ is a concatenation of $Q^{z_j}[R]$ for $R\in[k]^{k-j}$.
So the above equation holds for every choice of $R$.
Setting $R=\bar\pi$ and expanding with \Cref{itm:lem:general_cis_larger_k_node_6}, \Cref{itm:lem:general_cis_larger_k_leaf}, and \Cref{itm:lem:general_cis_larger_k_node_3} proves the claim.
\end{proof}

Given \Cref{clm:lem:general_cis_larger_k_2}, we analyze $c\cdot u$ as follows:
\begin{align*}
c\cdot u
&=\sum_\pi\sgn(\pi)
\sum_{z_1\in S^{z_0}_{\pi_1},\ldots,z_k\in S_{\pi_k}^{z_{k-1}}}c\left(\prod_{\ell\in[k]}\beta_{z_\ell}\right)v_{z_k}
\tag{by \Cref{eq:lem:general_cis_larger_k_2}}\\
&=\sum_\pi\sgn(\pi)
\sum_{\substack{z_1\in S^{z_0}_{\pi_1},\ldots,z_{j-1}\in S_{\pi_{j-1}}^{z_{j-2}}\\j:=\pi^{-1}(i^*)}}
\beta_{z_1}\cdots\beta_{z_{j-1}}
\sum_{z_j\in S_{i^*}^{z_{j-1}}}c
\sum_{z_{j+1}\in S_{\pi_{j+1}}^{z_j},\ldots,z_k\in S_{\pi_k}^{z_{k-1}}}
\left(\prod_{\ell=j}^k\beta_{z_\ell}\right)v_{z_k}\\
&=\sum_\pi\sgn(\pi)
\sum_{\substack{z_1\in S^{z_0}_{\pi_1},\ldots,z_{j-1}\in S_{\pi_{j-1}}^{z_{j-2}}\\j:=\pi^{-1}(i^*)}}\beta_{z_1}\ldots\beta_{z_{j-1}}\\
&\qquad\cdot\left(\sum_{z_j\in S_{i^*}^{z_{j-1}}}c-\sum_{i=0}^k\sum_{z_j\in S_i^{z_{j-1}}}\alpha_{z_j}\beta_{z_j}\right)\sum_{z_{j+1}\in S_{\pi_{j+1}}^{z_j},\ldots,z_k\in S_{\pi_k}^{z_{k-1}}}
\left(\prod_{\ell=j}^k\beta_{z_\ell}\right)v_{z_k}.
\tag{by \Cref{clm:lem:general_cis_larger_k_2}}
\end{align*}
By dividing the summation over $i$ into different cases, we have $c\cdot u=L^*-L_0-L_1$ where
\begin{equation*}
L^*=\sum_\pi\sgn(\pi)
\sum_{\substack{z_1\in S^{z_0}_{\pi_1},\ldots,z_k\in S_{\pi_k}^{z_{k-1}}\\j:=\pi^{-1}(i^*)}}
(c-\alpha_{z_j}\beta_{z_j})
\left(\prod_{\ell\in[k]}\beta_{z_\ell}\right)v_{z_k},
\tag{via $i=i^*$}
\end{equation*}
\begin{equation*}
L_0=\sum_\pi\sgn(\pi)
\sum_{\substack{z_1\in S^{z_0}_{\pi_1},\ldots,z_{j-1}\in S_{\pi_{j-1}}^{z_{j-2}}\\j:=\pi^{-1}(i^*)\\z_j\in S_0^{z_{j-1}}\\z_{j+1}\in S_{\pi_{j+1}}^{z_j},\ldots,z_k\in S_{\pi_k}^{z_{k-1}}}}
\alpha_{z_j}\beta_{z_j}
\left(\prod_{\ell\in[k]}\beta_{z_\ell}\right)v_{z_k},
\tag{via $i=0$}
\end{equation*}
and
\begin{equation*}
L_1=\sum_{i\in[k]\setminus\{i^*\}}\sum_\pi\sgn(\pi)
\sum_{\substack{z_1\in S^{z_0}_{\pi_1},\ldots,z_{j-1}\in S_{\pi_{j-1}}^{z_{j-2}}\\j:=\pi^{-1}(i^*)\\z_j\in S_i^{z_{j-1}}\\z_{j+1}\in S_{\pi_{j+1}}^{z_j},\ldots,z_k\in S_{\pi_k}^{z_{k-1}}}}
\alpha_{z_j}\beta_{z_j}
\left(\prod_{\ell\in[k]}\beta_{z_\ell}\right)v_{z_k}.
\tag{via $i\notin\{0,i^*\}$}
\end{equation*}

We use the following \Cref{clm:lem:general_cis_larger_k_3} to show that $c\cdot u$ is an $A_0$-nontrivial $H'$-sum.

\begin{claim}\label{clm:lem:general_cis_larger_k_3}
$L^*,L_0,L_1$ are disjoint vector sums (i.e., they sum over disjoint sets of leaves of $\calT$).
Moreover, $L^*$ is an $H'$-sum; $L_0$ is an $A_0$-nontrivial $H'$-sum; and $L_1$ is an $H'$-sum.
\end{claim}
\begin{proof}
We first simplify $L_1$.
Let $\Lambda\subseteq[k]^k$ be the set of $\sigma\in[k]^k$ such that $\sigma$ contains no $i^*$ and exactly $k-1$ distinct elements, i.e., $\{\sigma_1,\ldots,\sigma_k\}=[k]\setminus\{i^*\}$.
Then there are distinct $j:=j_\sigma\in[k]$ and $j':=j'_\sigma\in[k]$ such that $\sigma_j=\sigma_{j'}$.
Define $\pi_\sigma$ to be $\sigma$ with $\sigma_j$ replaced by $i^*$; and define $\pi_\sigma'$ to be $\sigma$ with $\sigma_{j'}$ replaced by $i^*$
Then $\pi_\sigma,\pi'_\sigma\in\calS_k$ and enumerating $\pi\in\calS_k$ is equivalent as enumerating $\sigma$ and going over $\pi_\sigma,\pi_\sigma'$.
Therefore we can simplify $L_1$ as
\begin{align}
L_1
&=\sum_{\sigma\in\Lambda}\sum_{(j,\pi)\in\{(j_\sigma,\pi_\sigma),(j'_\sigma,\pi'_\sigma)\}}\sgn(\pi)\sum_{z_1\in S_{\sigma_1}^{z_0},\ldots,z_k\in S_{\sigma_k}^{z_{k-1}}}\alpha_{z_j}\beta_{z_j}\left(\prod_{\ell\in[k]}\beta_{z_\ell}\right)v_{z_k}
\tag{noticing $i=\sigma_j$}\\
&=\sum_{\sigma\in\Lambda}\sum_{z_1\in S_{\sigma_1}^{z_0},\ldots,z_k\in S_{\sigma_k}^{z_{k-1}}}\sgn(\pi_\sigma)\left(\alpha_{z_{j_\sigma}}\beta_{z_{j_\sigma}}-\alpha_{z_{j_\sigma'}}\beta_{z_{j_\sigma'}}\right)\left(\prod_{\ell\in[k]}\beta_{z_\ell}\right)v_{z_k},
\label{eq:lem:general_cis_larger_k_3}
\end{align}
where we used $\sgn(\pi_\sigma)=-\sgn(\pi'_\sigma)$ for the last equality.

We observe that $L^*$ traverses $\calT$ from root to leaf based on $\pi\in\calS_k\subsetneq[k]^k$; $L_0$ traverses based on $\pi\in\calS_k$ with $i^*$ replaced by $0$; and $L_1$, via \Cref{eq:lem:general_cis_larger_k_3}, traverses based on $\sigma\in\Lambda$.
Since the above criterion does not overlap, they are disjoint vector sums.

Finally we analyze the coefficients.
\begin{itemize}
\item For $L^*$, we recall that $c\in B_{i^*}$ and, by \Cref{itm:lem:general_cis_larger_k_node_3}, $\alpha_{z_j}\beta_{z_j}\in B_{i^*}$, which shows that $L^*$ is a $ B_{i^*}-B_{i^*}\subseteq H'$-sum.
\item For $L_0$, we notice that $\alpha_{z_j}\in H_0$.
Given other signs multiplied together, $L_0$ is a $(\pm H_0)\subseteq H'$-sum.
Note that $L_0$ is $A_0$-nontrivial by \Cref{itm:lem:general_cis_larger_k_node_2,itm:lem:general_cis_larger_k_node_4}.
\item For $L_1$ and any fixed $\sigma\in\Lambda$, let $i=\sigma_{j_{\sigma}}\in[k]\setminus\{i^*\}$.
Then by \Cref{itm:lem:general_cis_larger_k_node_3}, we know both $\alpha_{z_{j_\sigma}}\beta_{z_{j_\sigma}}$ and $\alpha_{z_{j_\sigma}}\beta_{z_{j_\sigma'}}$ are in $B_i$.
Hence $L_1$ is a $\bigcup_{i\in[k]\setminus\{i^*\}}(B_i-B_i)\subseteq H'$-sum.
\qedhere
\end{itemize}
\end{proof}
Finally we remark that, given $c$, we can efficiently express $c\cdot u$ as the target $A_0$-nontrivial $H'$-sum by computing $L^*,L_0$ directly with their definition and computing $L_1$ with its simplified form \Cref{eq:lem:general_cis_larger_k_3}.
\end{proof}

\begin{remark}\label{rmk:lem:general_cis_larger_k}
We note that in \Cref{itm:lem:general_cis_larger_k_node} when depth $\ell=1$, we do not have to use the $H$-zero-sum algorithm with specific $A$-nontriviality. Instead, it suffices to use any nontrivial $H$-zero-sum algorithm. This will still guarantee the nontriviality of $u$ in \Cref{eq:lem:general_cis_larger_k_2}. Since this saving does not affect the final bound much, we do not introduce further complications here.
\end{remark}

\subsection{From reducible vector to zero-sum} \label{sec:general_cis_reducible_to_zero-sum}

Similar to \Cref{cor:iterate_sis_halving}, we can improve zero-sum algorithm from reducible vectors.

\begin{lemma}\label{lem:general_cis_reducible_to_zero-sum}
Let $H\subseteq\F_q$ and let $\emptyset\ne A\subseteq\F_q\setminus\{0\}$.
Suppose there is a deterministic algorithm $\calA$, given $m(n)$ vectors in $\F_q^n$ and running in time $T(n)$, that outputs an $s(n)$-sparse $A$-nontrivial $H$-zero-sum.

Let $H_0,H_1,\ldots,H_k$ be a nonempty disjoint partition of $H$ where $k\ge1$.
Assume $A_i:=(\pm A)\cap H_i$ is nonempty for all $i=0,\ldots,k$.
For $i\in[k]$, let $B_i\subseteq\F_q$ be arbitrary such that $H_i\subseteq\pm B_i$.
Define
$$
H'=(\pm H_0)\cup\bigcup_{i\in[k]}(B_i-B_i).
$$
Define $m_1=1$ and, for $\ell=2,3,\ldots,k+1$,
\begin{align*}
m_\ell
=\prod_{j=1}^{\ell-1}s(k^{k-j}n)+\sum_{i=1}^{\ell-1}\left(m(k^{k-i}n)-s(k^{k-i}n)\right)\prod_{j=1}^{i-1}s(k^{k-j}n).
\end{align*}
Then there is a deterministic algorithm $\calA'$, given $m(n)\cdot m_{k+1}$ vectors in $\F_q^n$ and running in time $T(n)+m(n)\cdot T'$, that outputs an $\left(s(n)\cdot s'\right)$-sparse $A_0$-nontrivial $H'$-zero-sum, where
$$
s'=\prod_{\ell=1}^ks(k^{k-\ell}n)
\qquad\text{and}\qquad
T'=\sum_{\ell=1}^km_\ell\cdot T(k^{k-\ell}n)+\poly(m_{k+1},n,k,\log q).
$$
\end{lemma}
\begin{proof}
We apply \Cref{lem:general_cis_larger_k-simple} or \Cref{lem:general_cis_larger_k} with $\calA$ to obtain an algorithm $\calB$ that constructs an $s'$-sparse $A_0$-nontrivial $(H_1\cup\cdots\cup H_k\to H')$-reducible vector for $m_{k+1}$ input vectors.
We run $\calB$ on $m=m(n)$ disjoint batches of input vectors to produce reducible vectors $u^{(1)},\ldots,u^{(m)}$.
Then we apply $\calA$ on $u^{(1)},\ldots,u^{(m)}$ to obtain an $A$-nontrivial $H$-zero-sum $\beta_1u^{(1)}+\cdots+\beta_mu^{(m)}=0$.
For each $j\in[m]$, if $\beta_j\in H_i$ for some $i\in[k]$, then we replace $\beta_ju^{(j)}$ with an $A_0$-nontrivial $H'$-sum by the reducibility of $u^{(j)}$.
Crucially since $\beta_1u^{(1)}+\cdots+\beta_mu^{(m)}=0$ is $A$-nontrivial and $\pm A$ hits some (in fact, every) $H_i,i\in[k]$, the above substitution always happens, which means the final $H'$-zero-sum is $A_0$-nontrivial.

The total number of vectors used is $m(n)$ batches of $m_{k+1}$ vectors; the runtime is mainly $m(n)$ executions of $\calB$ and one execution of $\calA$; and the sparsity is $s(n)\cdot s'$ where $s(n)$ is the sparsity of the $\beta_j$'s and $s'$ is the sparsity of each nonzero $\beta_ju^{(j)}$ (and its substitution).
\end{proof}

\begin{remark}\label{rmk:general_cis_reducible_to_zero-sum}
If we do not need to enforce the specific $A_0$-nontriviality in \Cref{lem:general_cis_reducible_to_zero-sum}, then we do not have to use the $A$-nontrivial $H$-zero-sum algorithm to combine the $A_0$-nontrivial $(H_1\cup\cdots\cup H_k\to H')$-reducible vectors. This still ensures a nontrivial $H'$-zero-sum in the end. However the saving does not affect the final bound much.
\end{remark}

\subsection{Improved \texorpdfstring{$\SIS$}{SIS-inf} algorithms}\label{sec:general_cis_sis}

The new reduction allows us to improve the bounds for the $\SIS$ problem and prove \Cref{thm:SIS}.

The previous halving trick corresponds to the case of $k=1$ and we will need to iterate many times to get solutions with smaller weights.
Now with the ability to handle general $k$, we do not have to iterate and can finish in one shot.

The following \Cref{thm:general_cis_sis} removes the assumption of $k$ being an integer power of $2$ in \Cref{thm:halving_sis} and proves \Cref{thm:SIS} immediately.

\begin{theorem}\label{thm:general_cis_sis}
Let $q\ge5$ be a prime and $2\le k\le\lfloor q/2\rfloor$ be an integer.
Given
$$
(k-1)^{(k-1)(k-2)/2}(n+k)^k
$$
vectors in $\F_q^n$, we can find in deterministic $\poly(m,k,\log q)$ time a nontrivial $(\pm\lfloor q/(2k)\rfloor)$-zero-sum.
\end{theorem}
\begin{proof}
Define $B_0,\ldots,B_{k-1}\subsetneq\{1,\ldots,\lfloor q/2\rfloor\}$ be contiguous intervals of length at most $\lfloor q/(2k)\rfloor+1$, where $B_0=\{1,\ldots,\lfloor q/(2k)\rfloor\}$.
Let $H_0=(\pm B_0)\cup\{0\}$ and $H_i=\pm B_i$ for $1\le i\le k-1$.
Since $H_0=\{-\lfloor q/(2k)\rfloor,\ldots,\lfloor q/(2k)\rfloor\}$, it is equivalent to find a nontrivial $H_0$-zero-sum.

Observe that $H_0,\ldots,H_{k-1}$ form a nonempty disjoint partition of $H:=\F_q$.
Let $A'\subseteq\F_q$ be of size $k$ such that $A'$ contains exactly one element of $H_i$ for each $0\le i\le k-1$.
By \Cref{lem:general_cis_starting_point} with $r=1$, given $m(n)$ vectors in $\F_q^n$, we can efficiently find an $s(n)$-sparse $A'$-nontrivial $\F_q$-zero-sum where $m(n)=s(n)=n+k$.
Let this be the algorithm $\calA$ for \Cref{lem:general_cis_reducible_to_zero-sum}.

We apply \Cref{lem:general_cis_reducible_to_zero-sum} with $H,H_0,\ldots,H_{k-1}$ and $B_1,\ldots,B_{k-1}$. Since each $B_i-B_i\subseteq H_0$, we have $H'=H_0$. Hence we can efficiently compute a nontrivial $H_0$-zero-sum given $\bar m$ input vectors, where
\begin{align*}
\bar m
&=(n+k)\cdot\prod_{j=1}^{k-1}\left((k-1)^{k-1-j}n+k\right)\\
&\le(n+k)^k\prod_{j=1}^{k-1}(k-1)^{k-1-j}=(k-1)^{(k-1)(k-2)/2}(n+k)^k.
\tag*{\qedhere}
\end{align*}
\end{proof}

\subsection{Improved \texorpdfstring{$\F_q^n$-\SubsetSum}{Fq-Subset-Sum} algorithms}\label{sec:general_cis_subset-sum}

Recall that in \Cref{sec:avg_subset-sum} we use worst-case $\SIS$ algorithms to derive average-case \SubsetSum algorithms.
Since we now have better $\SIS$ algorithms, it is natural to expect a similar improvement for \SubsetSum.

This is formalized as the following \Cref{thm:general_cis_subset-sum}, which should be compared with \Cref{thm:subset-sum_avg}.
Note that \Cref{thm:FpSS} follows immediately from \Cref{thm:general_cis_subset-sum}.

\begin{theorem}\label{thm:general_cis_subset-sum}
Let $q\ge5$ be a prime and define $k=\lfloor(q+3)/4\rfloor$. Let $\eps\in(0,1]$ be arbitrary.
Given
$$
m\ge(k-1)^{(k-1)(k-2)}(n+k)^{2k}\cdot O(\log(q)\log(1/\eps))=q^{O(q^2)}\cdot n^{2k}\cdot\log(1/\eps)
$$
uniform random vectors in $\F_q^n$, with probability at least $1-\eps$ we can find in deterministic $\poly(m)$ time a nontrivial $(0,1)$-zero-sum.
\end{theorem}
\begin{proof}
Note that our choice of $k$ satisfies $k>q/4$ and thus $\lfloor q/(2k)\rfloor=1$.
Then the analysis is identical to the proof of \Cref{thm:subset-sum_avg}, except that we use \Cref{thm:general_cis_sis} as the $(\pm1)$-zero-sum algorithm.
\end{proof}

It will also be useful to generalize $\{0,1\}$ to handle all sets of size $2$. This can be done with another application of \Cref{lem:cis_worst-avg}.

\begin{theorem}\label{thm:general_cis_size-two}
Let $q\ge5$ be a prime and define $k=\lfloor(q+3)/4\rfloor$. 
Let $\eps\in(0,1]$ be arbitrary and let $A\subseteq\F_q$ be of size $2$.
Given
$$
m\ge q^{O(q^2)}\cdot n^{2k}\cdot\log^2(1/\eps)
$$
uniform random vectors in $\F_q^n$, with probability at least $1-\eps$ we can find in deterministic $\poly(m)$ time a nontrivial $A$-zero-sum.
\end{theorem}
\begin{proof}
Observe that $A=a\{0,1\}+b$ for some $a\in\F_q\setminus\{0\}$ and $b\in\F_q$.
Let $\calA$ be the deterministic $(0,1)$-zero-sum algorithm in \Cref{thm:general_cis_subset-sum} and we apply \Cref{lem:cis_worst-avg} to obtain nontrivial $A$-zero-sum from nontrivial $(0,1)$-zero-sums.

While \Cref{lem:cis_worst-avg} is stated for deterministic worst-case algorithm $\calA$, it is also easy to see that it works for deterministic average-case algorithm $\calA$, which is what we have here.
The only change will be to ensure the success of individual runs of $\calA$.
This leads to the following set-up: we pick $d=\lceil\log(2q/\eps)\rceil$ in \Cref{lem:cis_worst-avg} and run $d$ independent $\calA$'s. We provide $\bar m=q^{O(q^2)}\cdot n^{2k}\cdot\log(2d/\eps)$ vectors for each individual $\calA$ to produce a nontrivial $(0,1)$-zero-sum.
By our parameter choice, each $\calA$ succeeds with probability $1-\eps/(2d)$; and, upon the success of all $\calA$'s executions, the final conversion to $A$-zero-sum succeeds with probability $1-\eps/2$.
By a union bound, the overall success probability is $1-\eps$ over the uniform random inputs while our algorithm is deterministic.
The total number of input vectors is $1+d\cdot\bar m=q^{O(q^2)}\cdot n^{2k}\cdot\log^2(1/\eps)$ as claimed.
\end{proof}

\subsection{Improved \texorpdfstring{$\CIS$}{CIS} algorithms}\label{sec:general_cis_cis}

Now we are ready to improve the simple $\CIS$ algorithm in \Cref{sec:simple_cis} and prove \Cref{thm:CIS_intro}.

Recall that in the $\CIS$ problem, we are given an arbitrary set $A$ of allowed coefficients. For convenience, one may think of $A$ as fixed in the problem description. If $A$ is also taken as an input of the problem (say, by explicit list), our reductions in this subsection still work and will have a mild $\poly(q)$ time overhead to find appropriate APs in $A$.

To handle such generality of $A$, in \Cref{lem:general_cis_larger_k} we will set $B_1,\ldots,B_k$ as singleton sets to ensure $B_i-B_i=\{0\}$ for all $i\in[k]$. This is formalized in the following \Cref{thm:general_cis_centered}.

\begin{theorem}\label{thm:general_cis_centered}
Let $q\ge5$ be a prime. 
Let $1\le a_1<\cdots<a_k\le\lfloor q/2\rfloor$ be arbitrary where $1\le k<\lfloor q/2\rfloor$.
Define $\bar A=\{\pm a_1,\ldots,\pm a_k\}$ and $A=\F_q\setminus\bar A$.
Then given 
$$
m\ge k^{k(k-1)/2}(n+k+1)^{k+1}
$$
vectors in $\F_q^n$, we can find in deterministic $\poly(m,\log q)$ time a nontrivial $A$-zero-sum.
\end{theorem}
\begin{proof}
Let $a_0\in A\setminus\{0\}$ be arbitrary. Define $A'=\{a_0,a_1,\ldots,a_k\}$. 
By \Cref{lem:general_cis_starting_point} with $r=1$, given $m(n)$ vectors in $\F_q^n$, we can efficiently find an $s(n)$-sparse $A'$-nontrivial $\F_q$-zero-sum where $m(n)=s(n)=n+k+1$.
Let this be the algorithm $\calA$ for \Cref{lem:general_cis_reducible_to_zero-sum}.

Define $H=\F_q$, $H_0=A$, and $H_i=\{\pm a_i\},B_i=\{a_i\}$ for $i\in[k]$. Then apply \Cref{lem:general_cis_reducible_to_zero-sum}, where $H'=A$.
Therefore we can efficiently compute a nontrivial $A$-zero-sum given $\bar m$ input vectors, where
\begin{align*}
\bar m
&=(n+k+1)\cdot\prod_{j=1}^k\left(k^{k-j}n+k+1\right)\\
&\le(n+k+1)^{k+1}\cdot\prod_{j=1}^kk^{k-j}
=k^{k(k-1)/2}(n+k+1)^{k+1}.
\tag*{\qedhere}
\end{align*}
\end{proof}

Using \Cref{lem:cis_worst-avg}, we can shift and dilate $A$ in \Cref{thm:general_cis_centered} to handle more general cases.

\begin{theorem}\label{thm:general_cis_paired}
Let $q\ge5$ be a prime.
Let $1\le a_1<\cdots<a_k\le\lfloor q/2\rfloor$ be arbitrary where $1\le k<\lfloor q/2\rfloor$.
Define $\bar A=\{\pm a_1,\ldots,\pm a_k\}$ and $A=\F_q\setminus\bar A$.

Fix arbitrary $a\in\F_q\setminus\{0\},b\in\F_q$ and define $B=aA+b$.
Let $\eps\in(0,1]$ be arbitrary.
Then given 
$$
m\ge1+k^{k(k-1)/2}(n+k+1)^{k+1}\cdot\lceil\log(q/\eps)\rceil
$$
uniform random vectors in $\F_q^n$, with probability at least $1-\eps$ we can find in deterministic $\poly(m,q)$ time a nontrivial $B$-zero-sum.
\end{theorem}
\begin{proof}
Follows immediately by combining \Cref{thm:general_cis_centered} and \Cref{lem:cis_worst-avg}.
\end{proof}

Observe that in \Cref{thm:general_cis_paired} we pay a cost of roughly $n^{k+1}$ to remove $2k$ elements, though these elements need to be paired up. This already hints an improvement over \Cref{thm:CLZCIS}, as they pay a cost of roughly $n^{2k+1}$ to remove only $2k$ elements.
In general, whenever we have two paired elements that need to be discarded, we only incur a linear cost, in contrast to the quadratic cost in \cite{chen2022quantum}.

By \Cref{thm:general_cis_paired}, we look into the following arithmetic combinatorial structure: up to some additive shift, we want to find some $-x,x\in\bar A$ while guaranteeing some $-y,0,y\in A$ to ensure $k<\lfloor q/2\rfloor$.
This leads to the following result, the proof of which is deferred to \Cref{app:arith_comb}.

\begin{restatable}{fact}{fctsimpleAPandantipodalhole}\label{fct:simple_3AP_and_antipodal_hole}
Let $q\ge3$ be a prime.
Let $A\subseteq\F_q$ and define $c=|\F_q\setminus A|$.
There exist $x,y,z\in\F_q$ such that the following holds.
\begin{enumerate}
\item\label{itm:fct:simple_3AP_and_antipodal_hole_1}
If $2\le c<q$, then $x\ne0$, $z\in A$, and $z-x,z+x\notin A$.
\item\label{itm:fct:simple_3AP_and_antipodal_hole_2}
In addition to the conditions in \Cref{itm:fct:simple_3AP_and_antipodal_hole_1}, if $c<(q+1)/2$, then $y\ne0$ and $z-y,z+y\in A$.
\end{enumerate}
\end{restatable}

The following fact completes the edge case in \Cref{itm:fct:simple_3AP_and_antipodal_hole_2} of \Cref{fct:simple_3AP_and_antipodal_hole}. Its proof is also presented in \Cref{app:arith_comb}.
A similar result over the integers is due to Erd{\H o}s and Tur{\'a}n \cite{erdos1936some}.

\begin{restatable}{fact}{fctmiddleAP}\label{fct:middle_3AP}
Let $q\ge11$ be a prime.
Let $A\subseteq\F_q$ and define $c=|\F_q\setminus A|$.
Assume $c=(q+1)/2$.
Then $A$ contains a $3$-AP.
\end{restatable}

The above results fall into the richer literature of quantitative Szemer\'edi's theorem. Using more advanced techniques, the bound in \Cref{fct:middle_3AP} can be strengthened to $c\le(1-o(1))\cdot q$ and $3$-AP can be extended into longer APs. Numerous works are along this line and we refer interested readers to \cite{kelley2023strong,leng2024improved} for recent breakthroughs. 

We emphasize that we choose to work with the simpler but worse bounds because (1) they are sufficient for our purposes, (2) they admit simpler proofs that we can present for self-containedness, and (3) they hold for small primes with explicit constants.

Now we use \Cref{fct:simple_3AP_and_antipodal_hole} and \Cref{fct:middle_3AP}, together with algorithms in \Cref{sec:general_cis_sis} and \Cref{sec:general_cis_subset-sum}, to prove the following \Cref{thm:general_cis_cis}, from which \Cref{thm:CIS_intro} follows immediately.

\begin{theorem}\label{thm:general_cis_cis}
Let $q\ge3$ be a prime.
Let $1\le c\le q-2$ be an integer and let $B\subseteq\F_q$ be arbitrary of size $q-c$.
Let $\eps\in(0,1]$ be arbitrary.
Then given $m$ uniform random vectors in $\F_q^n$, with probability at least $1-\eps$ we can find in deterministic $\poly(m,q)$ time a nontrivial $B$-zero-sum, where
\begin{itemize}
\item $m\ge n^{c+1}\cdot O\left(\log(q/\eps)\log(1/\eps)\right)$ if $c=1$ or ($c=3$ and $q=5$).
\item $m\ge n^c\cdot c^{O(c^2)}\cdot\log(q/\eps)\log(1/\eps)$ if otherwise, i.e., ($c\ge2$ and $q>5$) or ($c=2$ and $q=5$).
\end{itemize}
\end{theorem}
\begin{proof}
Let $\bar B\subseteq\F_q$, which has size $1\le c\le q-2$.
Then we have the following cases. 

\paragraph{The $q=3$ case.}
Note that $c=1$ and $|B|=2$, which is the (shifted) \FSubsetSum problem. Hence we use \Cref{thm:intro_F3-subset-sum} and \Cref{lem:cis_worst-avg} to handle this case. The number of input vectors is
$$
O(n^2)\cdot\lceil\log(3/\eps)\rceil=n^{c+1}\cdot O(\log(1/\eps)).
$$

\paragraph{The $c=1$ and $q\ge5$ case.}
Define $A=\{-\lfloor q/2\rfloor+1,\ldots,\lfloor q/2\rfloor-1\}$.
Under an additive shift $b\in\F_q$, $A+b$ is contained in $B$.
Hence we can apply \Cref{thm:general_cis_paired} with $k=1$ to obtain a nontrivial $A+b\subseteq B$-zero-sum.
The number of input vectors is
$$
1+(n+2)^2\lceil\log(q/\eps)\rceil=n^{c+1}\cdot O\left(\log(q/\eps)\right).
$$

\paragraph{The $2\le c\le(q-1)/2$ case.}
By \Cref{fct:simple_3AP_and_antipodal_hole} we find $z\in\F_q$ and $x,y\in\F_q\setminus\{0\}$ such that $z-x,z+x\notin B$ and $z-y,z,z+y\in B$.
Define $A=B-z$ and $\bar A=\F_q\setminus A$. Then $-x,x\in\bar A$ and $-y,0,y\in A$.
This means that we can find $1\le a_1<\cdots<a_k\le\lfloor q/2\rfloor$ such that $\bar A\subseteq\{\pm a_1,\ldots,\pm a_k\}$; in particular, $1\le k\le c-1<\lfloor q/2\rfloor$ since $\pm x\in\bar A$.
Hence we can apply \Cref{thm:general_cis_paired} with $k$ to obtain a nontrivial $A+b\subseteq B$-zero-sum.
The number of input vectors is
$$
1+k^{k(k+1)/2}(n+2)^{k+1}\cdot\lceil\log(q/\eps)\rceil\le c^{O(c^2)}\cdot n^c\cdot\log(q/\eps).
$$

\paragraph{The $(q+1)/2\le c\le q-2$ and $q\equiv3\pmod4$ case.}
In this case, we pick an arbitrary $A\subseteq B$ of size $2$ and apply \Cref{thm:general_cis_size-two} to obtain a nontrivial $A\subseteq B$-zero-sum.
The number of input vectors is
$$
q^{O(q^2)}\cdot n^{2\lfloor(q+3)/4\rfloor}\cdot\log^2(1/\eps)\le c^{O(c^2)}\cdot n^c\cdot\log^2(1/\eps),
$$
where we compute $2\lfloor(q+3)/4\rfloor=(q+1)/2\le c$.

\paragraph{The $(q+3)/2\le c\le q-2$ and $q\equiv1\pmod4$ case.}
Here we still reduce $B$ into a set of size $2$ and obtain a same bound as above.
The number of input vectors is
$$
q^{O(q^2)}\cdot n^{2\lfloor(q+3)/4\rfloor}\cdot\log^2(1/\eps)\le c^{O(c^2)}\cdot n^c\cdot\log^2(1/\eps),
$$
where we compute $2\lfloor(q+3)/4\rfloor=(q+3)/2\le c$.

\paragraph{The $c=(q+1)/2$ and $q\ge11$ case.}
Let $A=\{0,\pm1\}$. By \Cref{fct:middle_3AP}, we can find a $3$-AP in $B$. Hence for some $a\in\F_q\setminus\{0\}$ and $b\in\F_q$, we have $aA+b\subseteq B$.
By applying \Cref{thm:general_cis_sis} with $k=\lfloor(q+3)/4\rfloor$, we have $\lfloor q/(2k)\rfloor=1$ and thus we can efficiently find a nontrivial $A$-zero-sum given $q^{O(q^2)}n^{\lfloor(q+3)/4\rfloor}$ vectors.
Then by \Cref{lem:cis_worst-avg}, we can find nontrivial $aA+b\subseteq B$-zero-sums with an extra blowup of $O(\log(q/\eps))$.
Hence the number of input vectors is
$$
q^{O(q^2)}n^{\lfloor(q+3)/4\rfloor}\cdot O(\log(q/\eps))\le c^{O(c^2)}\cdot n^c\cdot\log(1/\eps),
$$
where we compute $\lfloor(q+3)/4\rfloor\le(q+1)/2=c$.

\paragraph{The $c=(q+1)/2$ and $q=5$ case.}
Now $c=3$ and $|B|=q-c=2$. By \Cref{thm:general_cis_size-two}, the number of input vectors is
$$
q^{O(q^2)}\cdot n^{2\lfloor(q+3)/4\rfloor}\cdot\log^2(1/\eps)=n^{c+1}\cdot O\left(\log^2(1/\eps)\right),
$$
where we compute $2\lfloor(q+3)/4\rfloor=4=c+1$.
\end{proof}

\subsection{Further optimizations}\label{sec:general_cis_sis_optimize}

Here we discuss extra tricks that can improve the sample complexity $m$ of the results in this section.

\paragraph{Exploring sparsity.}
For clean presentation, we apply \Cref{lem:general_cis_reducible_to_zero-sum} with the simplest choice $r=1$ throughout the section. However it is easy to see that a better choice would be $r\approx\log_q(n)$ when $q$ is relatively small compared to $n$. Then the bound in \Cref{lem:general_cis_starting_point} would be $s\approx(1-q^{-1})n$, which becomes a multiplicative saving of roughly $(1-q^{-1})^k$ for $m$ in \Cref{lem:general_cis_reducible_to_zero-sum} and all later applications.

\paragraph{Dimension reduction.}
Since we apply \Cref{lem:general_cis_reducible_to_zero-sum} directly with $H=\F_q$ through the section, the dimension reduction idea from \Cref{thm:f3-zero-sum-less-weak} applies.
This can save another multiplicative factor of $1/k$ in \Cref{thm:general_cis_sis}, \Cref{thm:general_cis_subset-sum}, and \Cref{thm:general_cis_size-two}; and a factor of $1/(k+1)$ in \Cref{thm:general_cis_centered} and \Cref{thm:general_cis_paired}.

In a bit more detail, if $H=\F_q$, then in \Cref{lem:general_cis_reducible_to_zero-sum} each time we construct a reducible vector, we can safely project later vectors into its complementary space. This in general saves a factor of $1/(\ell+1)$ as $\sum_{i\le n}i^\ell\approx n^{\ell+1}/(\ell+1)$, where each $i^\ell$ is roughly the cost of \Cref{lem:general_cis_reducible_to_zero-sum} with $\ell+1$ partitioned sets for input vectors in a subspace of dimension $i$.

\paragraph{Iterative application.}
All our results in this section use only a single shot of \Cref{lem:general_cis_reducible_to_zero-sum}. It is natural to wonder if some iterative application similar to \Cref{thm:halving_sis} will be beneficial.
Consider the $\SIS$ problem where we want to find a nontrivial $(\pm\lfloor q/(2k)\rfloor)$-zero-sum.
If we apply \Cref{lem:general_cis_reducible_to_zero-sum} once, then we obtain \Cref{thm:general_cis_sis} and a bound of $m\approx k^{k^2/2}n^k$ samples.
Now we show how to improve it when $k$ can be factored as a product of small integers.

Say $k=q_1\cdots q_t$ where each $q_i\ge2$ is an integer. Assume the field size $q$ is much smaller than the vector dimension $n$ for simplicity.
Define $k_i=q_1\cdots q_i$ where $k_0=1$ and $k_t=k$.
We will iteratively apply \Cref{lem:general_cis_reducible_to_zero-sum} to $(\pm\lfloor q/(2k_i)\rfloor)$-zero-sum algorithms from $(\pm\lfloor q/(2k_{i-1})\rfloor)$-zero-sum algorithms. While there is some nontriviality subtlety that needs to be preserved in order to repeatedly apply \Cref{lem:general_cis_reducible_to_zero-sum}, it will just be a minor factor due to $q\ll n$.

The starting point is the $(\pm\lfloor q/2\rfloor)$-zero-sum algorithm $\calA_0$ from \Cref{lem:general_cis_starting_point} that uses roughly $m_0(n)=n$ input vectors.
Then we apply \Cref{lem:general_cis_reducible_to_zero-sum} to obtain a $(\pm\lfloor q/(2k_1)\rfloor)$-zero-sum algorithm that uses roughly $m_1(n)=q_1^{q_1^2/2}n^{q_1}$ input vectors.
Then we apply \Cref{lem:general_cis_reducible_to_zero-sum} to divide $q_2$ which ends up needing roughly
$$
m_2(n)\approx\prod_{j=1}^{q_2}m_1\left(q_2^{q_2-j}n\right)\approx\prod_{j=1}^{q_2}\left(q_2^{q_2-j}\cdot q_1^{q_1^2/2}n^{q_1}\right)\approx q_1^{q_1^2q_2/2}\cdot q_2^{q_2^2/2}\cdot n^{q_1q_2}
$$
input vectors.
Similar calculation shows that in the end, we can find $(\pm\lfloor q/(2k)\rfloor)$-zero-sum when the number of input vectors is 
$$
m_t(n)\approx n^{q_1\cdots q_t}\cdot\prod_{i=1}^tq_i^{(q_i^2/2)\cdot\prod_{j=i+1}^tq_j}.
$$
If each $q_i$ is roughly the same around $r:=k^{1/t}$, then it simplifies to
$$
m_t(n)\approx n^k\cdot\prod_{i=1}^tr^{r^{t-i+2}/2}
=n^k\cdot\left(r^{1/2}\right)^{\sum_{i=1}^tr^{t+2-i}}
\approx n^k\cdot\left(r^{1/2}\right)^{r^{t+1}}
=n^k\cdot r^{kr/2}.
$$
Therefore, if $r$ is small (say, constant), it improves the $k^{k^2/2}$ prefactor significantly to $O(1)^k$; and the worst case is $k$ being prime and $t=1$, for which it falls back to the one-shot case unless we are willing to sacrifice dependence on $n$.


\section{Discussion}\label{sec:discussion}

In this section, we describe improvements, generalizations, and open directions; and also discuss prior works related to our paper.

\subsection{On our results}\label{sec:discussion_our}

\paragraph{Targeted sum and more general $\CIS$.}
While our work focuses on the zero-sum case in both the worst-case and average-case setting, it is meaningful to consider the case where we want to target some vector other than $0$.
This is a ``closest-vector problem (CVP)'' analogue, thinking of $\SIS$ as a ``shortest-vector problem (SVP)''.  The targeted version is at least as hard as the zero-sum version, although one might have an intuition it is ``not too much'' harder --- at least in the average-case.  There \emph{are} some differences though; for example, whereas $\F_q^n$-\SubsetSum is guaranteed to have a solution whenever $m > (q-1)n$ (see \Cref{fct:KO}), this is not true in the targeted case.
For the average case where input vectors are uniformly at random, our algorithm works with minor change.

In addition, one may want to study the generalization of $\CIS$ where the $i$th coefficient has its own allowed set $A_i\subseteq\F_q$. Since our arithmetic combinatorial results depend only on the size of the allowed sets, they also work here; same for the reduction to handle translation and dilation (\Cref{lem:cis_worst-avg}). Therefore our results on the $\CIS$ problem hold with the same bound in this more general setting.

\paragraph{Better runtime.}
We can make the runtime of our algorithms more explicit. Recall that our algorithms rely heavily on the weight reduction, either through the halving trick in \Cref{sec:halving_trick} or through the more general trick in \Cref{sec:general_cis}. It is easy to see that the reduction itself is computationally efficient in linear time, and the main runtime bottleneck lies in the base case where we find a nontrivial (sparse) linear dependence with every field element as an allowed coefficient. Such base case algorithms are presented in \Cref{lem:f3-zero-sum-zeros}, \Cref{lem:sis_halving_sparse}, and \Cref{lem:general_cis_starting_point}; and they boil down to the following simple linear algebra problems: find basis vectors in input vectors and make some invertible linear transform.

Such problems naturally benefit from fast matrix multiplication algorithms.
Let $2\le\omega<2.372$ denote a constant such that $n \times n$ matrix multiplication (and inversion) over $\F_q$ can be done in deterministic classical time $n^{\omega} \polylog(n,q)$.
See \cite{alman2025more} and references within for latest results.
The following \Cref{fct:fast_basis} provides what we need, and its proof can be found in \Cref{app:fast_linear_algebra}.

\begin{restatable}{fact}{fctfastbasis}\label{fct:fast_basis}
There is a deterministic algorithm, given any $m\ge n$ vectors in $\F_q^n$ and running in time $n^{\omega-1}m\cdot\polylog(n,q)$, that outputs a maximal linearly independent set among the vectors.
\end{restatable}

Using \Cref{fct:fast_basis}, the runtime of our $\SIS$ algorithms can be optimized to $n^{\omega-1}m\cdot\polylog(n,q)$, and the runtime of our \SubsetSum and $\CIS$ algorithms is $\left(n^{\omega-1}m+\poly(q)\right)\cdot\polylog(n,q)$. By way of comparison, the CLZ quantum algorithm takes time at least $m^{\omega}$; so, for example, our \FSubsetSum algorithm takes $n^{\omega+1}$ classical time, whereas CLZ's is $n^{2\omega}$ quantum time.

\paragraph{Better sample complexity.}
The sample complexity $m$ is the main object we optimize in the paper. In \Cref{sec:general_cis}, we prioritized clarity over optimality and obtained simple bounds that have targeted dependence on $n$ but worse dependence in terms of other parameter like $q$ (field size) or $k$ (the partition number in \Cref{lem:general_cis_reducible_to_zero-sum}). In \Cref{sec:general_cis_sis_optimize}, we discussed tricks that bring in additional savings on existing bounds in various settings. However it remains a challenge to significantly improve the dependence on $q$ and $k$ for all cases; and it is even less clear how to (or whether it is possible to) improve the dependence on $n$.

For concreteness, we ask if the $k^{O(k^2)}n^k$ bound in \Cref{thm:general_cis_sis} can be improved to $O(n^k)$ (which would match \Cref{thm:halving_sis} but hold for general $k$) or even $n^{o(k)}$.

Also for the $\F_q^n$-\SubsetSum problem, we ask if it is possible to solve worst-case instances with $O_q\left(n^{O(q)}\right)$ input vectors, which, if true, would improve the $O_q\left(n^{O(q\log q)}\right)$ bound in \cite{II24} and match the average-case complexity.
Finally for the \FSubsetSum problem, we ask if it is easy to solve with only $m=o(n^2)$ input vectors.

\paragraph{General finite field.}
While all the finite fields in our paper are prime field, most of our results work for general finite fields. However it may not be necessary to do so as we can always view the field elements as a vector of number in a prime field.
In a bit more detail, let $p$ be a prime and $t\ge1$ be an integer such that $q=p^t$.
We can view $\F_q$ as a vector space $\F_p^t$; thus dimension-$n$ input vectors $v_1,\ldots,v_m$ over field $\F_q$ can be converted into dimension-$tn$ input vectors $v_1',\ldots,v_m'$ over field $\F_p$.
In addition, any zero-sum of $v_1',\ldots,v_m'$ is a zero-sum of $v_1,\ldots,v_m$ with the same set of coefficients.
Hence we may alternatively just use the zero-sum algorithms over $\F_p$.

\paragraph{Ring of integers modulo $q$.}
Another meaningful generalization is to consider input vectors with entries in $\Z/q\Z$ where $q\ge2$ is an integer. This reduces to $\F_q$ if $q$ is a prime. For composite number $q$, our algorithms are still applicable with minor tweak. Take results in \Cref{sec:general_cis} for example. The basic zero-sum algorithm in \Cref{sec:general_cis_basic_zero-sum} is fine as we can do Gaussian elimination with the subtractive Euclidean algorithm.
Then the reductions in \Cref{sec:general_cis_zero-sum_to_reducible} and \Cref{sec:general_cis_reducible_to_zero-sum} are already good as only addition and subtraction are involved.

On the other hand, tweaking our algorithms may not be necessary. There is a simple reduction, which is essentially \cite[Appendix A]{chen2022quantum} and is credited to Regev, that converts a $(\pm h_1)$-zero-sum algorithm $\calA_1$ over $\Z/q_1\Z$ and a $(\pm h_2)$-zero-sum algorithm $\calA_2$ over $\Z/q_2\Z$ into a $(\pm h_1h_2)$-zero-sum algorithm over $\Z/(q_1q_2\Z)$.
The observation is that $v\equiv0\pmod{q_1q_2}$ is equivalent to (1) $v\equiv0\pmod{q_1}$ and (2) $(v/q_1)\equiv0\pmod{q_2}$.
Therefore we can first run $\calA_1$ to obtain zero-sums modulo $q_1$ (i.e., satisfying condition (1)); then run $\calA_2$ on those zero-sums divided by $q_1$ to obtain zero-sums modulo $q_2$ (i.e., satisfying condition (2)).
The final sample complexity is the sample complexities multiplied.

\subsection{Works by Imran, Ivanyos, Sanselme, and Santha}\label{sec:discussion_ivanyos}

The most comparable papers to the present one are a line of work by Ivanyos and coauthors~\cite{ISS12,IS17,II24}.  These papers are all motivated by the Chevalley--Warning theorem and by designing \emph{quantum} algorithms for versions of the Hidden Subgroup Problem; yet, all contain intermediate results that can be viewed as classical algorithms for $\SIS$ or $\CIS$.  

The work of Ivanyos--Sanselme--Santha~\cite{ISS12} considered the problem of finding nontrivial solutions to (worst-case) systems of ``diagonal'' quadratic equations over $\F_q$, i.e., equations of the form $\sum_j h_{j} x_j^2 = 0$.  When there are $m$ equations, this can be regarded as the $\CIS$ problem with allowed set $A = \{a^2 : a \in \F_q\}$ the quadratic residues.  They gave an efficient algorithm whenever $m \geq (n+1)(n+2)/2$.  As noted in \Cref{sec:f3}, when $q = 3$, this problem is equivalent to \FSubsetSum.

The followup work of Ivanyos and Santha~\cite{IS17} generalized the problem and algorithm to solving systems of diagonal $d$th-power equations; this is equivalent to $\CIS$ with $A = \{a^d : a \in \F_q\}$.  They gave a $\poly(m, \log q)$-time classical algorithm whenever 
$$
m \geq d^{d(d-1)\lceil\log_2(d+1)\rceil/2} (n+1)^{d \lceil \log_2(d+1)\rceil} \approx d^{d^2 \log d} n^{d \log d}.
$$
They also observed that taking $d = q-1$ gives an $\F_q^n$-\SubsetSum algorithm.

Finally, Imran and Ivanyos~\cite{II24} also investigated $\F_q^n$-\SubsetSum and slightly sharpened the result of~\cite{IS17}, finding solutions in $\poly(m)$ time whenever 
$$
m \geq q^{C q \log^2 q} \cdot n^{\lfloor q/2 \rfloor \cdot \lceil \log_2 q \rceil} \approx q^{q \log^2 q} \cdot n^{q \log q}.
$$
A central ingredient in their method is what we call the ``halving trick''.  The base version of this trick lets them solve $\SIS$ with $s = \lfloor q/4 \rfloor$ whenever $m \geq C q n^2$, and they also iterate it to achieve smaller $s$.  In  \Cref{sec:halving_trick}, we give an improvement on the halving trick that eliminates the $q$-factor in the $m$-dependence: it achieves $s = \lfloor q/4 \rfloor$ whenever $m \geq C n^2$.  This means the halving trick can be employed effectively even when $q$ is exponentially large in~$n$.  One can get a lot of mileage out of the halving trick, but our strongest results (\Cref{sec:cis}) are ultimately obtained by abandoning it in favor of a more sophisticated generalization in \Cref{sec:general_cis}.

\subsection{Quantum algorithms from Regev's reduction}\label{sec:discussion_regev}

Aside from the work of Chen, Liu, and Zhandry \cite{chen2022quantum}, at least two other candidate exponential quantum speedups have been derived via Regev's framework of converting an $A$-$\CIS$-type problem (usually in its equivalent dual version) to a quantum decoding~\cite{chailloux2024quantum} problem.  These are: the Yamakawa--Zhandry paper~\cite{yamakawa2024verifiable}, which gives a provable exponential quantum speedup in the random oracle query model; and, the DQI paper of Jordan et~al.~\cite{jordan2024optimization}, which gives a seeming exponential quantum speedup for the ``OPI'' problem of fitting a low-degree univariate polynomial over~$\F_q$ to a range of points (see also the followup~\cite{chailloux2024quantum}).  

It is interesting to understand what aspects of these problems cause them to resist our dequantization efforts.  The fact that both problems are ``worst-case/structured'' --- meaning that the input codes are not assumed to be chosen randomly --- does not necessarily pose a problem for dequantization, as many of our algorithms work in the worst case.
The Yamakawa--Zhandry problem has $q$ exponentially large as a function of~$n$, but this too does not seem to be an inherent problem for classical algorithms, as our algorithms mostly have $\log q$ dependence.  The fact that the set~$A$ in Yamakawa--Zhandry has no structure --- indeed, is random --- certainly makes it harder to dequantize.
That said, we do not know how to dequantize even the OPI problem when the set $A$ is assumed to be an interval of width~$q/2$, the most structured possible~$A$.
It seems that the biggest difficulty in dequantizing both Yamakawa--Zhandry and the DQI algorithm for OPI is the fact that $m = O(n)$ for these problems; whereas, our efficient classical algorithms do not seem to be able to get off the ground unless $m \geq \Omega(n^2)$.

\subsection{\texorpdfstring{$\SIS$}{SIS-inf}-related problems}\label{sec:sis_related}

Here we mention some problems that are either equivalent to $\SIS$, or very nearly equivalent.

\paragraph{Versus traditional \SubsetSum.}  $\F_q^n$-\SubsetSum is a variant of the traditional \SubsetSum problem (with target~$0$), one of the most canonical $\mathsf{NP}$-complete problems.  In fact, the textbook reduction from (say) \textsc{1-in-3-Sat} to \SubsetSum works out much more cleanly for $\F_q^n$-\SubsetSum, for any $q > 2$, since no tricks are needed to prevent carries.\footnote{The reduction is as follows. Given an $m$-clause $n$-variate \textsc{1-in-3-Sat} instance, produce $2n+1$ vectors in dimension $n+m$. For each literal $\ell_i$, make a vector with $1$ in the $i$th coordinate and in the $(i+j)$th coordinate whenever $\ell_i$ satisfies the $j$th clause.  Finally, include the vector $-(1, \dots, 1)$. Correctness uses $1 \not \in \{0,2, 3\}$ mod~$q$.}  We also remark that over integers, the $\{\pm 1\}$-$\CIS$ problem is known as \textsc{Partitioning} (one of Garey and Johnson's 6 canonical $\mathsf{NP}$-complete problems) or \textsc{Pigeonhole-Equal-Sums}.  It is a ``total'' problem by the Pigeonhole Principle whenever $2^m > q^n$. This implies that $m > (\log_2 q) n$, and hence the problem is in the $\mathsf{TFNP}$ class $\mathsf{PPP}$.
We also mention that a variant of $\SIS$ is proven to be $\mathsf{PPP}$-complete \cite{SZZ18}.

\paragraph{The equivalent ``dual'' problem: LWE.}  We recall the \emph{equivalent}, linear algebraic dual problem to $\SIS$.
Given input $H \in \F_q^{n \times m}$, the $\SIS$ problem may be described as seeking a (nonzero) \emph{codeword} $x$, satisfying $\|x\|_\infty \leq s$, in the $\F_q$-linear code whose parity-check matrix is~$H$.  The problem is easily interreducible to the version in which the code is presented by a \emph{generator matrix} $G \in \F_q^{k \times m}$, $k \coloneqq m-n$.  Now the task may be described as seeking an unconstrained $z \in \F_q^k$ such that $zG$ is nonzero but ``short'', $\|zG\|_\infty \leq s$.

This problem looks even more familiar if one starts from the more general targeted version of $\SIS$.  Then, given $H$ and target $t$, one is seeking a short codeword~$x$ in the \emph{affine} code defined by $Hx = t$. In the generator matrix formulation, this is seeking an unconstrained $z\in \F_q^k$ such that $\|zG - b\|_\infty \leq s$ for an arbitrary $b$ satisfying $Hb = t$.  Equivalently, this is seeking a mod-$q$ solution $z \in \F_q^k$ to the $m$ ``noisy equations'' 
\begin{equation}
    g_i \cdot z \approx b_i,
\end{equation}
where $\approx$ denotes the $\ell_\infty$ norm of their difference is at most $s$.
In this (equivalent) formulation, the problem strongly resembles the Learning With Errors (LWE) problem. The main differences are:
\begin{itemize}
    \item LWE usually assumes one is in the \emph{planted random} case, where a ``secret'' solution $z^*$ is chosen at random, then the vectors $g_i$ are chosen uniformly at random, and finally each $b_i$ is taken to be the planted value $g_i \cdot z^*$ with random noise of ``width'' $s$. The noise is typically discrete-Gaussian of standard deviation~$s$, but uniform on $-s, \dots, s$ is also reasonable.
    \item In LWE, the parameters are usually chosen so that $m = k^C$ for some (possibly large) constant~$C$.  In the (equivalent) targeted $\SIS$ problem, this corresponds to ``$m$'' (the number of input vectors) being only slightly larger than ``$n$'' (the dimension): $m = n + n^{1/C}$.
\end{itemize}
Recall that when the input is uniformly random (not planted), there is unlikely to be a solution unless $m \geq (\log_2 q) n$.  Thus in the ``LWE regime'' of $m = n + n^{1/C}$, it only makes sense to study the refutation problem or the planted problem. In the present paper, our $\SIS$ results are only for $m$ at least $\Omega(n^2)$, so we have nothing to say about LWE.

\paragraph{LIN-SAT.} 
The equivalent dual version of $\CIS$ is a common viewpoint for the recent quantum algorithms using Regev's reduction.  In this viewpoint, the input is an $\F_q$-linear code $\calC \subseteq \F_q^m$ of dimension $k= m-n$ together with a subset $A \subseteq \F_q$. The task is to find a codeword $y \in \calC$ such that $y_i \in A$ for all $i \in [m]$.  
In the targeted and even more general version, this would mean $y_i \in b_i + A_i$ for all~$i$, where $b$ is also part of the input.

This form of the problem was called \textsc{LIN-SAT} in \cite{jordan2024optimization}.
This viewpoint is adopted in the Yamakawa--Zhandry problem~\cite{yamakawa2024verifiable}, where $q$ is exponentially large in~$m$, $\calC$ is an efficiently decodeable, list-recoverable code (specifically, a folded Reed--Solomon code), and $A \subseteq \F_q$ is a \emph{random} set of cardinality roughly $q/2$.
It is also the viewpoint in the Optimal Polynomial Interpolation (OPI) problem from~\cite{jordan2024optimization}, where we work in the targeted version, $q \approx m$, $k \approx m/10$, and $\calC$ is the Reed--Solomon code.

\paragraph{The $q = 3$ case.}  
We remark that the $q = 3$ case of $\{\pm 1\}$-$\CIS$ is particularly attractive. In its original, primal version, it is asking to partition a sequence of $m$ vectors from $\F_q^n$ into two subsequences of equal sum.  In its dual version, it is asking to find a codeword in a given $\F_3$-linear code with all symbols in $\{\pm 1\}$, i.e., at maximal Hamming distance~$m$ from the all-zero codeword.  This ``ternary syndrome decoding with maximal weight'' problem was carefully studied in the context of the Wave cryptosystem~\cite{debris2019wave}. To write the problem ``LWE''-style, it is equivalent to solving a system of \emph{disequations} over~$\F_3$, 
\begin{equation}
    g_i \cdot s \neq 0.
\end{equation}
This problem of \emph{learning from disequations}, which also makes sense over fields larger than~$\F_q$, has been studied in several prior works, dating back to Friedl et~al.~\cite{FIMSS02}, and developed by Ivanyos~\cite{Iva07}, Arora--Ge~\cite{arora2011new}, and Ivanyos--Prakash--Santha~\cite{IPS18}.

\subsection{Zero-sum theory}\label{sec:related_zero-sum}

The following arguably surprising fact dates to the late 1960s.

\begin{fact}[\cite{Ols69,van69}]\label{fct:KO}
    Let $q$ be a prime.
    Given sequence $h_1, \dots, h_m \in \F_q^n$, there is always a nonempty subsequence summing to~$0$ provided $m > (q-1)n$ (and this bound is tight).
\end{fact}

In other words, for $m > (q-1)n$ the $\F_q^n$-\SubsetSum problem (and hence $\SIS$ for any~$s$) is ``total'': there is always a solution.  Unlike other similar constraint satisfaction problems, this puts $\F_q^n$-\SubsetSum into the ``$\mathsf{TFNP}$'' framework; in particular, for $m > (q-1)n$ the task is $\mathsf{PPA}_q$~\cite{GKSZ20,SZZ18}.

\Cref{fct:KO} holds more generally for $p$-groups, and is one of the seminal results in the field of \emph{zero-sum theory}.  For a survey of this research area, which investigates \SubsetSum and related problems like $\SIS$ with $s = 1$ over abelian groups, see e.g.,~\cite{GG06}.

The original proof of \Cref{fct:KO} was via an exponential-time algorithm.  These days it is common to derive it from the Chevalley--Warning Theorem, or more generally Combinatorial Nullstellensatz.  For interested readers, we record here a streamlined proof.

\begin{proof}[Proof of \Cref{fct:KO}]
    Recall that $x \mapsto x^{q-1}$ maps zero/nonzero to $0/1$ in $\F_q$. Thus
    we need to find nonzero $x \in \F_q^m$ such that $\text{EQ}_i(x) \coloneqq \sum_{j=1}^m h_{j}[i]x_j^{q-1} = 0$ for $i = 1\dots n$, where $h_j[i]$ is the $i$th coordinate of $h_j$.
    We encode the conditions with polynomials:
    \begin{equation}
        \text{SAT}(x) \coloneqq (1-\text{EQ}_1(x)^{q-1})(1-\text{EQ}_2(x)^{q-1}) \cdots (1-\text{EQ}_n(x)^{q-1})  \in \{0,1\}
    \end{equation}
    signifies whether $x$ gives a zero-sum, and 
    \begin{equation}
        Z(x) \coloneqq (1-x_1^{q-1})(1-x_2^{q-1})\cdots (1-x_m^{q-1}) \in \{0,1\}
    \end{equation}
    detects whether $x = 0$.
    Thus we seek any $x \in \F_q^m$ such that $\text{OK}(x) \coloneqq \text{SAT}(x) - Z(x) \neq 0$.
    Recall that every function $f : \F_q^m \to \F_q$ has a unique representation as a polynomial with individual degree at most $q-1$, obtained by taking any interpolating polynomial and reducing via $x_i^q \mapsto x_i$.  But even after reduction, $\text{OK}(x)$ will have a term $\pm x_1^{q-1} \cdots x_m^{q-1}$ of degree $(q-1)m$ coming from $Z(x)$. This cannot be canceled out from $\text{SAT}(x)$ because $\text{SAT}$ has degree at most $(q-1)^2n$, which is strictly less than $(q-1)m$ by assumption.
    Thus $\text{OK} : \F_q^m \to \F_q$ has unique polynomial representation which is nonzero, and is thus a nonzero function, i.e., there is at least one $x \in \F_q^m$ with $\text{OK}(x) \neq 0$.
\end{proof}

\subsection{Post-quantum cryptosystem}\label{sec:discussion_crypto}

We discuss two recent proposals for NIST post-quantum cryptosystems whose security relies on the hardness of variants of $\SIS$.  We stress that both cryptosystems have $m \approx 2n$, whereas our algorithms require $m \geq \Omega(n^2)$.  

The \emph{Wave} signature scheme~\cite{debris2019wave,BCCCDGKLNSST22} is essentially based on the assumed hardness of average-case $\{\pm 1\}$-$\CIS$ over $\F_3^n$, which is equivalent in hardness to average-case \FSubsetSum.  More precisely, it is based on a version with: (i)~a planted trapdoor; (ii)~a relaxation that only ``most'' coordinates of the solution need be in $\{\pm 1\}$. Nevertheless, the most stringent cryptanalysis of the system~\cite{BCDL19} focused on the basic average-case $\CIS$ problem.  The recommended parameter setting for optimal security was roughly $m \approx 2n$, and \cite{BCDL19} presented convincing evidence that a break required $2^{\Omega(n)}$ time.  For $m \gg n$, a Blum--Kalai--Wasserman/Wagner-style algorithm~\cite{blum2003noise,wagner2002generalized,ducas2025wagner} will solve the problem in time $\exp(O(n/\log(m/n)))$.
But it is interesting to speculate if this can be improved for $m = n^{c}$ where $1 < c < 2$.
We know from our \Cref{thm:intro_F3-subset-sum} that the task becomes easy once $m \ge n^2/3$.

The \cite{DKLLSSS18,LDKLSSS19} CRYSTALS-Dilithium signature scheme was selected for standardization by NIST's post-quantum cryptography initiative.  One of the key hardness assumptions underlying its security is ``Module-$\SIS$'', which is a variant of $\SIS$.  A concrete, comparable version of $\SIS$ that is presumed intractable has $q \approx 2^{23}$, $s \sim q/8$, and $m \sim 1.9 n$, with $n = 1280$.  Our $\SIS$ algorithms are successful at handling $q \gg n$, but require $m \geq Cn^4$ in order to achieve $s \sim q/8$.  Again, it would be interesting to investigate subexponential-time tradeoffs that interpolate between the polynomial-time-solvable case of $m \approx n^4$ and the presumed exponential-time case of $m \approx n$.

\section*{Acknowledgments}

We are very grateful to David Gosset for his contributions to the early stages of this work. RO thanks Siddhartha Jain for helpful discussions.  KW thanks Ce Jin, Qipeng Liu, and Hongxun Wu for relevant references. We acknowledge the use of Gemini and ChatGPT to search the literature and suggest proof strategies. We thank anonymous reviewers for helpful comments.

KW is supported by the National Science Foundation under Grant No. DMS-2424441, and by the IAS School of Mathematics.

\bibliographystyle{alphaurl} 
\bibliography{ref}

\appendix

\section{\textit{Ad hoc} arithmetic combinatorics}\label{app:arith_comb}

We prove the arithmetic combinatorial results here. 
While significantly better bounds are known \cite{kelley2023strong,leng2024improved} using more sophisticated techniques, we choose to provide what we need with minimal efforts as this is not the focus of our work.

We first prove \Cref{fct:lev_long_AP}.

\fctlevlongAP*

\begin{proof}
Let $\bar A=\F_q\setminus A=\{a_1,\ldots,a_c\}$ where $c=|\bar A|\le\log_4(q+2)$.
By translation, we assume $0\notin\bar A$.
Partition $\F_q\setminus\{0\}$ into $4$ contiguous intervals $I_1,I_2,I_3,I_4$, each of length at most $\lceil(q-1)/4\rceil$.

For each $s\in\F_q\setminus\{0\}$, define $v^{(s)}\in\{0,1,2,3\}^c$ where the $i$th entry of $v^{(s)}$ indicates the interval that $s\cdot a_i$ lies in, i.e., $v^{(s)}[i]=b\in\{0,1,2,3\}$ if $s\cdot a_i\in I_b$.
Now we divide into the following cases.
\begin{itemize}
\item If for some $s$ and $b\in\{0,1,2\}$ we have $v^{(s)}[i]=b$ for all $i\in[c]$, then $s\bar A$ is contained in $I_b$, which is an interval of length $|I_b|\le\lceil(q-1)/4\rceil\le(q-1)/2$. This means $sA$ contains an interval of length at least $(q+1)/2$, which is an AP of the same length in $A$ after dilation.
\item The above cases correspond to $4$ patterns in $\{0,1,2,3\}^c$. If none of them appear, it remains $4^c-4$ possible patterns. On the other hand, we have $q-1$ choices of $s$. Since $q-1\ge4^c-3$, there exist distinct $s,s'\in\F_q\setminus\{0\}$ such that $v^{(s)}=v^{(s')}$.
Then for $\bar s=s-s'\ne0$, each $\bar s\cdot a_i=s\cdot a_i-s'\cdot a_i\in(I_b-I_b)\setminus\{0\}$ for some $b\in\{1,2,3,4\}$.
Since each $I_b$ is an interval of length at most $\lceil(q-1)/4\rceil$, every element in $I_b-I_b$ has absolute value at most $\lceil(q-1)/4\rceil-1\le(q-1)/4-1/2$ since $q-1\pmod4\in\{0,2\}$.
This means $\bar sA$ contains the interval of length $q-2\left((q-1)/4-1/2\right)-1=(q+1)/2$, which is an AP of the same length in $A$ after dilation.
\qedhere
\end{itemize}
\end{proof}

Now we turn to \Cref{fct:simple_3AP_and_antipodal_hole} and \Cref{fct:middle_3AP}.

\fctsimpleAPandantipodalhole*

\begin{proof}
Define $\bar A=\F_q\setminus A$.
For distinct $u,v\in\bar A$, let $z(u,v)=\frac{u+v}2$, which is well-defined as $q\ge3$.
Consider $Z=\{z(u,v)\colon u,v\in\bar A\text{ and }u\ne v\}$. We will show that $Z\cap A\ne\emptyset$.

For each $z\in\bar A$, we have
\begin{equation}\label{eq:fct:simple_3AP_and_antipodal_hole_1}
|\{(u,v)\colon u,v\in\bar A\text{ and }u\ne v\text{ and }z(u,v)=z\}|\le2\lfloor(c-1)/2\rfloor,
\end{equation}
since $u,z(u,v),v$ is a nontrivial $3$-AP in $\bar A$.
Hence 
$$
P:=\sum_{z\in\bar A}|\{(u,v)\colon u,v\in\bar A\text{ and }u\ne v\text{ and }z(u,v)=z\}|\le2c\lfloor(c-1)/2\rfloor.
$$
On the other hand, we have
$$
Q:=\sum_{z\in Z}|\{(u,v)\colon u,v\in\bar A\text{ and }u\ne v\text{ and }z(u,v)=z\}|
=|\{(u,v)\colon u,v\in\bar A\text{ and }u\ne v\}|=c(c-1).
$$
Assume towards the contradiction that $Z\setminus\bar A=Z\cap A=\emptyset$.
Then $Z\subseteq\bar A$ and hence $P\ge Q$.
\begin{itemize}
\item If $c$ is even, then $Q=c(c-1)>2c\lfloor(c-1)/2\rfloor=P$. A contradiction.
\item If $c$ is odd, then $Q=c(c-1)=2c\lfloor(c-1)/2\rfloor=P$. This means every inequality in \Cref{eq:fct:simple_3AP_and_antipodal_hole_1} is an equality. That is, for every $z\in\bar A$, the set $\bar A\setminus\{z\}$ is formed by $(c-1)/2$ nontrivial antipodal pairs $\{u,v\}$ with the common center $z$.
Note that each such a pair $\{u,v\}$ satisfies $u+v=2z$.
Hence $\sum_{u\in\bar A}u=(1+2\cdot(c-1)/2)\cdot z=c\cdot z$ holds for any $z\in\bar A$, which means $c\cdot z=c\cdot z'$ for any $z,z'\in\bar A$.
Since $2\le|\bar A|=c<q$, this is impossible.
\end{itemize}
Now we fix an arbitrary $z\in Z\cap A$ and arbitrary $u\ne v\in\bar A$ satisfying $z(u,v)=z$. Define $x=z-u$. Then \Cref{itm:fct:simple_3AP_and_antipodal_hole_1} immediately holds.

To find $y$ satisfying \Cref{itm:fct:simple_3AP_and_antipodal_hole_2}, we observe that $\{z-x,z+x\}\subseteq\bar A\subseteq\F_q\setminus\{z\}$ by \Cref{eq:fct:simple_3AP_and_antipodal_hole_1}.
Since $q\ge3$, we know that $\F_q\setminus\{z-x,z,z+x\}$ is a disjoint partition of $(q-3)/2$ antipodal pairs $\{z-y,z+y\}$ for $y\notin\{0,\pm x\}$.
As $\left|\bar A\setminus\{z-x,z+x\}\right|=c-2<(q-3)/2$, there exists a pair $\{z-y,z+y\}$ for some $y\notin\{0,\pm x\}$ that is not contained in $\bar A$, i.e., $z-y,z+y\in A$ as desired.
\end{proof}

To handle the edge case that $c=(q+1)/2$ and prove \Cref{fct:middle_3AP}, we need the following simple fact.

\begin{fact}\label{fct:3AP_in_pm4}
Let $c_{-4},c_{-3},\ldots,c_4\in\{0,1\}$ be arbitrary satisfying $c_0=0$ and $c_i+c_{-i}=1$ for $i=1,2,3,4$.
Then there always exist $-4\le j<k<\ell\le4$ such that $j+\ell=2k$ and $c_j=c_k=c_\ell=0$.
\end{fact}
\begin{proof}
Assume towards the contradiction that such $j,k,\ell$ do not exist. 

By symmetry, we fix $c_1=0,c_{-1}=1$.
To avoid $c_0=c_1=c_2=0$, we must have $c_{-2}=0,c_2=1$.
To avoid $c_{-2}=c_1=c_4=0$, we now must have $c_{-4}=0,c_4=1$.
However at this point we cannot avoid $c_{-4}=c_{-2}=c_0=0$, which is a contradiction.
\end{proof}

Now we prove \Cref{fct:middle_3AP}.

\fctmiddleAP*

\begin{proof}
Since $c=(q+1)/2$ and $q\ge11$, we know $2\le c<q$.
By \Cref{itm:fct:simple_3AP_and_antipodal_hole_1} of \Cref{fct:simple_3AP_and_antipodal_hole}, we can find $x\ne0$ and $z\in A$ such that $z-x,z+x\notin A$.

Define $m=\lfloor q/2\rfloor$ and $t=m/x\in\F_q$.
Consider $A'=t(A-z)$. Then $0\in A'$ and $\pm m\notin A'$.
In addition, any $3$-AP in $A'$ corresponds to a $3$-AP in $A$ by translation and dilation.
Hence it suffices to find a $3$-AP in $A'$.

Observe that $0\in A'\subseteq\F_q\setminus\{\pm m\}$.
Since $q>2$, we know that $\F_q\setminus\{\pm m\}$ is a disjoint partition of $(q-3)/2$ antipodal pairs $\{\pm y\}$ for $y\notin\{0,\pm m\}$.
Define $\bar{A'}=\F_q\setminus A'$.
As $\left|\bar{A'}\setminus\{\pm m\}\right|=c-2=(q-3)/2$, we run into one of the following two cases.
\begin{itemize}
\item There exists a pair $\{\pm y\}$ for some $y\notin\{0,\pm m\}$ that is not contained in $\bar{A'}$. 

Then we observe that $-y,0,y\in A'$ is a $3$-AP.
\item For every $y\notin\{0,\pm m\}$, exactly one element of the pair $\{\pm y\}$ is contained in $\bar{A'}$.

Since $q\ge11$ and $m=\lfloor q/2\rfloor\ge5$, we now focus on the range $i=-4,-3,\ldots,4\in\F_q$ of $9$ elements.
For each $i$ in the range, define $c_i=0$ if $i\in A'$; and $c_i=1$ if $i\in\bar{A'}$.
Then we have $c_0=0$ and $c_i+c_{-i}=1$.
By \Cref{fct:3AP_in_pm4}, there exist $-4\le j<k<\ell\le4$ such that $j+\ell=2k$ and $c_j=c_k=c_\ell=0$, which means $j,k,\ell\in A'$.
Hence $j,k,\ell$ is a $3$-AP in $A'$.
\end{itemize}
For both cases, we find a $3$-AP in $A'$ (and thus $A$), which completes the proof of \Cref{fct:middle_3AP}.
\end{proof}

\section{Fast basis search}\label{app:fast_linear_algebra}

We use fast matrix multiplication to obtain fast algorithms to find basis vectors and prove \Cref{fct:fast_basis}.

We will need the following subroutine.

\begin{fact}\label{fct:fast_basis_simple}
There is a deterministic algorithm, given any $1\le\ell\le2n$ vectors in $\F_q^n$ and running in time $\ell^{\omega-1}n\cdot\polylog(n,q)$, that outputs a maximal linearly independent set among the vectors.
\end{fact}
\begin{proof}
Let $T(\ell)$ be the runtime upper bound and apparently $T(1)=n\cdot\polylog(n,q)$ as claimed.

The algorithm proceeds in a recursive fashion for $\ell\ge2$: we partition the $\ell$ vectors into two parts with $\ell_1=\lfloor\ell/2\rfloor$ and $\ell_2=\lceil\ell/2\rceil$ vectors each, denoted $v_1,\ldots,v_{\ell_1}$ and $u_1,\ldots,u_{\ell_2}$ respectively.
We recursively find a maximal linearly independent set of size $t\le\ell_1$, stacked as $V\in\F^{n\times t}$, from the first part.
Then we project $u_1,\ldots,u_{\ell_2}$ onto the orthogonal space spanned by columns of $V$; and find a maximal linearly independent set there.
The final output is by merging all the vectors found above.

To analyze the runtime, recall that the projector onto the column space of $V$ is defined as $P=V(V^\top V)^{-1}V^\top$.
Define $U=\begin{bmatrix}u_1&\cdots&u_{\ell_2}\end{bmatrix}$.
Then the projection of $u_1,\ldots,u_{\ell_2}$ is equivalent to computing $U-PU=U-V(V^\top V)^{-1}V^\top U$.
Since $V\in\F^{n\times t},t\le\ell_1=\lfloor\ell/2\rfloor\le n$ and $U\in\F^{n\times\ell_2},\ell_2\le\lceil\ell/2\rceil\le n$, the computation of $A=V^\top V$ and $B=V^\top U$ takes time $\ell^{\omega-1}n\cdot\polylog(n,q)$ by partitioning $U,V$ into square blocks and combining $O(n/\ell)$ multiplications of $O(\ell)$ by $O(\ell)$ matrices.
Then $A^{-1}$ takes time $\ell^\omega\cdot\polylog(n,q)\le\ell^{\omega-1}n\cdot\polylog(n,q)$.
Finally $VA^{-1}B$ takes time $\ell^{\omega-1}n\cdot\polylog(n,q)$ by a similar block matrix multiplication fashion.
Hence the runtime $T(\ell)$ satisfies the following recursion
\begin{align*}
T(\ell)
&=\underbrace{T(\ell_1)}_\text{first recursion}+\underbrace{\ell^{\omega-1}n\cdot\polylog(n,q)}_\text{projection}+\underbrace{T(\ell_2)}_\text{second recursion}+\underbrace{\ell n\cdot\polylog(n,q)}_\text{merge}\\
&=\polylog(n,q)\cdot\left(\sum_{0\le i\le\log(\ell)}2^i\cdot(\ell/2^i)^{\omega-1}\cdot n\right)\\
&=\ell^{\omega-1}n\log(n)\cdot\polylog(n,q)=\ell^{\omega-1}n\cdot\polylog(n,q)
\tag{since $\omega\ge2$ and $\ell\le2n$}
\end{align*}
as desired.
\end{proof}

Now we are ready to prove \Cref{fct:fast_basis}.

\fctfastbasis*

\begin{proof}
We divide the $m$ input vectors into $k=\lceil m/n\rceil$ batches, where each batch contains at most $n$ vectors.
We maintain a set $S$ of linearly independent vectors and go over each batch to update $S$.
Since each batch has at most $n$ vectors and $|S|\le n$, each update on $S$ is equivalent to finding a basis for at most $2n$ vectors in $\F^n$, which has runtime $n^\omega\cdot\polylog(n,q)$ by \Cref{fct:fast_basis_simple}.
Hence the total runtime is $n^\omega k\cdot\polylog(n,q)=n^{\omega-1}m\cdot\polylog(n,q)$.
\end{proof}

\end{document}